\documentclass[a4paper,12pt]{article}
\usepackage{amsfonts}
\usepackage{amsfonts}
\usepackage{amsfonts}
\usepackage[top=0.75in,bottom=1in,left=0.75in,right=0.75in]{geometry}

\usepackage{amsmath,amssymb,amsthm}
\usepackage{graphicx}
\usepackage{enumerate}
\usepackage{subfig}
\usepackage{url}

\usepackage[dvipdfm]{hyperref}
\usepackage{algorithm, algorithmic}

\begin{document}
% don't want date printed
\date{}

%=============================================================================

\title{\LARGE\bf Mathematical Modeling of Competition in Sponsored Search Market}
% for single author (just remove % characters)
\author{Jian Liu \hspace{0.15in} Dah-Ming Chiu\\
        Department of Information Engineering\\
        The Chinese University of Hong Kong\\
        Email: \{lj008, dmchiu\}@ie.cuhk.edu.hk}

\maketitle \thispagestyle{empty}

\newtheorem{theo}{Theorem}
\newtheorem{prop}{Proposition}
\newtheorem{lem}{Lemma}
\theoremstyle{definition}
\newtheorem{defn}{Definition}
\newtheorem{expl}{Example}
\newtheorem{assumption}{Assumption}
\theoremstyle{remark}
\newtheorem{rem}{Remark}

\begin{abstract}
Sponsored search mechanisms have drawn much attention from both academic community and industry in
recent years since the seminal papers of \cite{SellingBillionsDollar} and \cite{PositionAuction}.
However, most of the existing literature concentrates on the mechanism design and analysis within the
scope of only \emph{one} search engine in the market. In this paper we propose a mathematical
framework for modeling the interaction of publishers, advertisers and end users in a
\emph{competitive} market. We first consider the monopoly market model and provide optimal solutions
for both \emph{ex ante} and \emph{ex post} cases, which represents the \emph{long-term} and
\emph{short-term} revenues of search engines respectively. We then analyze the strategic behaviors of
end users and advertisers under duopoly and prove the existence of equilibrium for both search
engines to co-exist from \emph{ex-post} perspective. To show the more general \emph{ex ante} results,
we carry out extensive simulations under different parameter settings. Our analysis and observation
in this work can provide useful insight in regulating the sponsored search market and protecting the
interests of advertisers and end users.
\end{abstract}

\section{Introduction}
Internet advertising has become a main source of revenue for primary search engines nowadays.
According to the newly-released report by Interactive Advertising Bureau and PricewaterhouseCoopers
\cite{IABreport}, Internet advertising in the United States reached \$22.7 billion in total revenue
for the full year of 2009, where sponsored search revenue accounted for 47 percent of the total
revenue.

A typical Internet search market consists of three parties: \emph{publishers} (i.e., search engines),
\emph{advertisers} and \emph{end users}. In the current age of information explosion, more and more
people rely on search engines to pin down their favored products or services. Whenever a query is
submitted to the engines by end users, their \emph{intents} or \emph{interests} can be potentially
captured by the engines through the inputted keywords. These intents of search users can then be sold
by search engines to companies who are interested in targeting their products to these users.
Nowadays, major search engine operators like Google, Yahoo! and Microsoft display advertisements in
the form of sponsored links, which appears alongside the algorithmic links (also known as organic
links) in the search results pages. For each keyword, there are usually more than one available
advertising slot in the search engine results page. How to effectively allocate these slots and
charge the advertisers have been studied and discussed extensively in recent years among people in
both academia and industry. Take Google's AdWords program for example. In this advertising program,
advertisers could choose multiple keywords they are interested in, and for each keyword indicate the
maximal willingness to pay for each click and the budget to spend over a period of time. Whenever a
user clicks on the sponsored link and is re-directed to the advertisers' site, certain payment is
charged by the program until the advertisers' budget is used up.

Most of the existing works focus on the interaction of the three parties within the scope of only
\emph{one} search engine's advertising system, and these results and suggestions from researchers did
greatly improve the efficiency of mechanism held in major search engine companies. For example, the
transition from generalized first price auction to generalized second price auction, from payment per
impression to payment per click, from bid-based ranking to quality-based ranking and so on
\cite{SponSearchHistory, SponSearchAuc}. However, considering there is usually more than one company
providing search service in the market, one natural question would be how the market would evolve
when there exists competition between \emph{multiple} search engines. In particular, will all users
and advertisers gradually concentrate to one leading engine or will the ``inferior'' engines still
earn enough profits to survive when competing with the leading one? What would be the consequences if
one search engine monopolizes the market? Should governments take anti-trust action against potential
cooperation among major search engines? These concerns arise from the current situation of high
levels of concentration in search engine market: Google is widely considered to possess the leading
technology and obtain the largest market shares in most countries and regions, followed by Yahoo! and
Microsoft Bing.

This paper aims to formulate a reasonable model to study the competition between two search engine
operators and help to address some of the intriguing questions mentioned above. We will consider a
three-stage dynamic game model. In stage I, the two operators' services determine how the market of
end users is split. In stage II, the two operators simultaneously determine their prices to
advertisers. In stage III, the advertisers choose the operator in which they can obtain highest
utility based on the announced prices in stage II. Each operator wants to maximize its revenue
subject to the competition for advertisers from the other operator.

The key contributions of the paper are:
\begin{itemize}
  \item To the best of our knowledge, this paper is among the first to model the \emph{comprehensive} interaction of \emph{heterogenous} publishers,
  advertisers and users in a \emph{competitive} environment. We elaborate the whole
  process of search engine competition in attracting both users and advertisers. The technical
  details applied in practical advertising systems, such as the values and budgets of advertisers,
  are taken into consideration explicitly in our formulation.
  \item We prove the existence of Nash Equilibrium prices when allowing advertisers participate in both
  advertising systems simultaneously. Moreover, we show the comparative results between the equilibrium
  prices and monopoly price when only one search engine exists in the market.
  \item We propose two distinct points of view, i.e., the \emph{ex ante} view and the \emph{ex post} view,
  to inspect the \emph{long-term} and \emph{short-term} revenues of search engines respectively.
  Furthermore, for the ex post case, we present the detailed algorithm for calculating the exact allocation
  results for advertisers. For the ex ante case, we can carry out the simulations to illustrate the comparative results of expected
  revenues and social welfare under competition and monopoly respectively.
\end{itemize}

The paper is organized as follows. Section \ref{sec:2:related_work} presents the related work.
Section \ref{sec:3:mono_model} formulates the monopoly model and provides solutions for both the ex
ante and ex post cases. Based on the monopoly formulation, we analyze the strategic behaviors of end
users and advertisers under duopoly and find the Nash Equilibrium prices of both search engines in
Section \ref{sec:4:duo_model}. In Section \ref{sec:5:simul} we compare the results under competition
and monopoly via simulation and reveal various factors that affect the revenues and social welfare.
We conclude in Section \ref{sec:6:conclu} and outline some future research directions.
%for readability, we defer the technical proofs of theorems to the Appendix.

\section{Related Work}\label{sec:2:related_work}
There are mainly two lines of research work in the sponsored search area. The mainstream of
literature focuses on the interaction between advertisers and search engines and aims to understand
and devise viable mechanism for the Internet advertising market. There is significant work on the
auction mechanism held by major search engines, starting from two seminal works of
\cite{SellingBillionsDollar} and \cite{PositionAuction} which independently investigated the
``generalized second-price'' (GSP) auction prevailing in major search engines such as Google and
Yahoo!. In \cite{Truth_Auc} the authors compared the ``direct ranking'' method by Overture with the
``revenue ranking'' method by Google and proposed a \emph{truthful} mechanism named as ``laddered
auction''. Considering the non-strategyproofness property of GSP mechanism, \cite{Vindictive}
analyzed one prevalent strategy of advertisers called ``vindictive bidding'' in real-world keyword
auction. \cite{Externality:1:CascadeModel} and \cite{Externality:2:MarkovianUser} relaxed the basic
assumption of separable click-through rate in \cite{Truth_Auc} and modeled the \emph{externality
effect} among advertisements which appeared in the same search page simultaneously.
\cite{Externality:3:EffectCompetingAd} proposed a new valuation model to absorb the adverse effect of
the competing advertisements on the advertiser's value per click. There is also an abundance of works
on proposing more \emph{expressive} but still \emph{scalable} mechanism for sponsored search such as
\cite{Contexts, CostConcise, ExpressScalableAuc, Externality:4:Expressive, Hybrid}. In particular, in
\cite{Externality:4:Expressive} the advertisers were allowed to submit a two-dimensional bid $(b,
b')$ where $b$ was the bid for exclusive display and $b'$ for sharing slots with other
advertisements. In \cite{Hybrid} the authors proposed a \emph{truthful} hybrid auctions where
advertisers can make a per-impression as well as per-click bid and showed that it can generate higher
revenue for search engine compared with the pure per-click scheme.

It's worth pointing out that a few works considered the practical situation where similar keyword
auctions are held simultaneously by \emph{multiple} search engines. For example, in
\cite{CompetingAdAuctions} the authors investigated the revenue properties of two search engines with
different click-through rates which competed for the same set of advertisers. The study in
\cite{CompetingKeywordAuctions} considered competition between two search engines which differs in
their ranking rules: one applied the direct ranking method we mentioned above, and the other applied
the revenue ranking method.\footnote{In the original paper of \cite{CompetingKeywordAuctions}, the
authors used the terms of ``price-only ranking rule'' and ``quality-adjusted ranking rule'' which has
the same meaning.} We assert that this assumption of search engine difference is unrealistic since
major engines tend to use the same policy which proves to work efficiently in practice and it is
unlikely that certain engine would switch back to the obsolete rules.\footnote{One typical example is
that, in May 2002 Google first introduced the revenue ranking approach which proved to be more
efficient. And then in 2007, Yahoo! switched from prior direct ranking to revenue ranking rule
similar to Google's \cite{SponSearchHistory}.} The Nash equilibrium solution in the former paper of
\cite{CompetingAdAuctions} is also not so practical since it requires advertisers to adopt certain
randomized strategy. It's very difficult for individual advertisers to implement such complex
strategies which would incur unnecessary maintenance cost.

The other line of work is developed mainly by economists to address the broad issues of search engine
competition from social welfare perspective. \cite{Rufus} introduced a quality choice game model
where end users choose the search engine with highest quality of search results, and showed that no
Nash equilibrium exists in this game. Based on this proposition, the author argues that the search
engine market would evolve towards monopoly in the absence of necessary regulatory interventions.
\cite{Beef_up} proposed a duopoly model which has some similarity to our formulation, however, as
many of the technical details of the practical advertising system are ignored, it is doubtful whether
this can serve as an accurate model to predict the outcome of search engine market. Similarly,
\cite{Rufus} faces the same problem that the vague description of participants' utility may not be
strong enough to support the predictive conclusions in the paper.

These two lines of important work have little intersections so far: the mainstream of work
concentrates on the technical progress in designing ``better'' advertising system, and the other line
usually involves less technical details (like the budgets of advertisers in practical advertising
system) and targets the macro-effect of competitive market. In view of this, we believe that a
\emph{comprehensive} study of the current search engine ecosystem in a \emph{competitive} way is
vital for addressing many of the unresolved issues in this thriving market. Our work manages to
narrow the gap between these two directions of research and makes some initial progress in this
direction. This observation helps differentiate our work from most of the existing literature.

\section{The Monopoly Market Model}\label{sec:3:mono_model}
In this section we consider the monopoly market first, which serves as a starting point for analyzing
the more general competition market. Suppose there is only one search engine in the market servicing
a fixed set of end users and providing advertising opportunity for a set of advertisers denoted by
$\mathcal {I}$ ($|\mathcal {I}|=m$). Assuming all users are homogeneous and each of them tends to
generate the same number of impressions (or clicks) for a particular keyword (query). Since we
assumes that the search engine owns a fixed number of users, it would be able to supply a fixed
number of attentions (in the form of impressions or clicks) for advertisers. For a given interval,
let the supply of attentions be $S$. Each advertiser $i\in\mathcal{I}$ has two private parameters:
\emph{value} $v_i$ denoting $i$'s maximal willingness to pay for each attention and \emph{budget}
$B_i$ in a given time interval (could be daily, weekly, monthly budget and so on). The search engine
needs to determine the optimal price per attention to maximize its revenue\footnote{In practice, the
optimal price is usually determined automatically by an auction mechanism. Specifically, this
automation process can be imagined as an ascending-bid auction \cite{AscendingAuc} where the
auctioneer (i.e., the search engine) iteratively raises the price until there is no excessive demand
than supply. Considering strategic issues, \cite{Truth_Auc_Budget} proposes an asymptotically
revenue-maximizing truthful mechanism. For simplicity of analysis, we ignore the detailed
implementation of auctions and assume the search engine can solve the revenue-maximizing problem
instantaneously.}:
\begin{equation*}
    R=p\cdot \min(S,D(p))=\min(p\cdot S, pD(p))
\end{equation*}
where $D(p)$ is the demand function over price $p$.

In the following analysis we consider this revenue maximization problem in two different
perspectives: the \emph{ex ante} perspective where the search engine only has an rough estimation to
the parameters of participating advertisers, and the \emph{ex post} perspective where the engine just
needs to make decision based on the \emph{submitted} parameters of advertisers. Although in practice,
the advertising systems do determine the prices only \emph{after} advertisers have submitted their
values and budgets, we assert that the ex ante view of revenue to be a natural fit for the search
engine's objective. This is because typically the interaction between search engine and advertisers
is not one-shot and would usually last for many rounds. Advertisers can actually adjust their
submitted parameters at any time to achieve better payoff. Thus the ex ante result could provide
valuable prediction of the \emph{long-term} revenue for the search engine, rather than the
\emph{short-term} profit from one particular instance of the ex post case.

\subsection{Ex ante case}
Assuming the search engine can have a rough estimation of the distribution of advertisers' values and
budgets. For simplicity, we only consider the scenario when the parameters are independent and
identically-distributed random variables. To be specific, suppose values are drawn independently from
distribution with density function $f(v)$ and CDF $F(v)$ over the range of $[\underline{v},
\overline{v}]$, and budgets are drawn independently from distribution with density function $g(B)$
and CDF $G(B)$.

After search engine announce the uniform price $p$, advertiser $i$ would make the deal if only the
value $v_i$ is larger than $p$. The quantity advertiser $i$ could purchase is constrained by the
budget $B_i$.

Therefore the expected aggregate demand under price $p$ from all advertisers would be:
\begin{equation*}
    D(p)=\sum_{i\in \mathcal I}\frac{E(B_i)}{p}\cdot\mbox{Prob}\{v_i>p\}=m\cdot \frac{E(B)}{p}[1-F(p)]
\end{equation*}
Rewrite it as:
\begin{equation}\label{eq:S_equal_D}
    p\cdot D(p)=m\cdot {E(B)}\cdot[1-F(p)]
\end{equation}
which is a non-increasing function over $p$.

We can use figure \ref{SE_Profit_over_Price_ExAnte} to illustrate the revenue of search engine over
price $p$.

\begin{figure}[ht]
  \centering
  \includegraphics[scale=0.75]{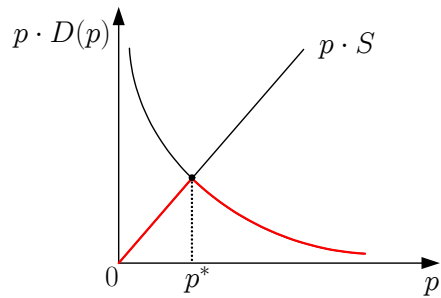}\\
  \caption{Search Engine Revenue Over Prices (Ex Ante)}\label{SE_Profit_over_Price_ExAnte}
\end{figure}

\begin{prop}\label{DemandEqualSupply}
The revenue $R$ is maximized when $S=D(p)$, i.e., when the demand equals the supply.
\end{prop}

Proposition \ref{DemandEqualSupply} can be proved by contradiction. If supply exceeds demand under
the current price, the search engine will cut down the price to achieve higher revenue (since
$R=p\cdot D(p)$ is non-increasing over price $p$); if demand exceeds supply, the search engine can
raise the price and reach a higher revenue (since $R=p\cdot S$ is monotonically increasing over price
$p$).

\begin{expl}
Assuming $v_i$ is drawn from uniform distribution with positive support on the interval
$[\underline{v},\overline{v}]$, where $0\leq \underline{v}< \overline{v}$ and $\Delta
v=\overline{v}-\underline{v}$. Then $1-F(p)=\frac{\overline{v}-p}{\Delta v}$. From $S=D(p)$ we have:

$$p^*=\frac{m\cdot E(B)\cdot \overline{v}}{m\cdot E(B)+S\cdot \Delta v}$$

The intuition is that the more demand (larger $m$) there is, the higher market clearing price would
be; and the more supply (larger $S$) there is, the lower market clearing price is.
\end{expl}

\subsection{Ex post case}
In practical search engine advertising system, advertisers need to submit their values and budgets to
the advertising system. The search engine therefore could determine the optimal price based on the ex
post variables.

Reorder the index of advertisers such that $v_j\leq v_{j+1}, j=1,\ldots,m-1$. Then the aggregate
demand can be written as:
\begin{equation*}
    D(p)=\sum_{i\in \mathcal I^+(p)}\frac{B_i}{p}
\end{equation*}
where we define the set:
\begin{equation*}
    {\mathcal I^+(p)}\triangleq \{i\in\mathcal I:v_i>p\}
\end{equation*}
Thus $p\cdot D(p)=\sum_{i\in \mathcal I^+(p)} B_i$ is a non-increasing function over $p$ since
${\mathcal I^+(p)}$ shrinks as price $p$ increases. By letting demand equal to supply, we have
\begin{equation*}%\label{eq:opt_price_ex_post}
    {p}(\mathcal{I})=\frac{\sum_{i\in\mathcal{I}^+({p})}B_i}{S}
\end{equation*}

Notice that the term of price appears in both sides of the equation. Thus in general we cannot derive
the closed-form solution for optimal price. Since $pD(p)$ is piece-wise constant and (weakly)
decreasing over $p$, we can illustrate the search engine revenue through examples in figure
\ref{fig:SE_Profit_over_Price_ExPost}. Here we assume there are four advertisers ordered such that
$v_1\leq v_2\leq v_3\leq v_4$ and initially when the price is zero, ${\mathcal I^+(p)}=\mathcal
I=\{1,2,3,4\}$. As the price exceeds $v_1$, advertiser $1$ would have no incentive to stay and
${\mathcal I^+(p)}$ becomes $\{2,3,4\}$. The crossing point of demand and supply shows that the
optimal price $p^*$ is located in $[v_1,v_2]$. To be more exact, $p^*(\mathcal{I})=(B_2+B_3+B_4)/S$.
In figure \ref{SE_Profit_over_Price_2_ExPost} we also show the other case when there is one
advertiser who is indifferent between participating and quitting the ad campaign since the optimal
price is equal to its value. In typical search engine systems like Google AdWords, after advertisers
input their maximal willingness to pay (i.e., their values) and budgets, the ad system would
automatically allocate attentions to advertisers as long as the current price doesn't exceed their
values and the budgets have not been exhausted yet. Thus here for ease of expression we can assume
that the indifferent advertiser would continue participating in the ad campaign under the budget
constraint. For example, as shown in figure \ref{SE_Profit_over_Price_2_ExPost}, the optimal price is
equal to $v_2$ (satisfying $B_3+B_4<v_2S<B_2+B_3+B_4$), advertiser 2 would consume the remaining
supply of $S-\frac{B_3+B_4}{v_2}$ and only spent $v_2S-B_3-B_4$ which is less than its budget $B_2$.

\begin{figure}[ht]
  \centering
  \subfloat[Determined Advertisers]{
    \label{SE_Profit_over_Price_ExPost}
    \begin{minipage}{0.4\textwidth}
        \centering
        \includegraphics[scale=0.75]{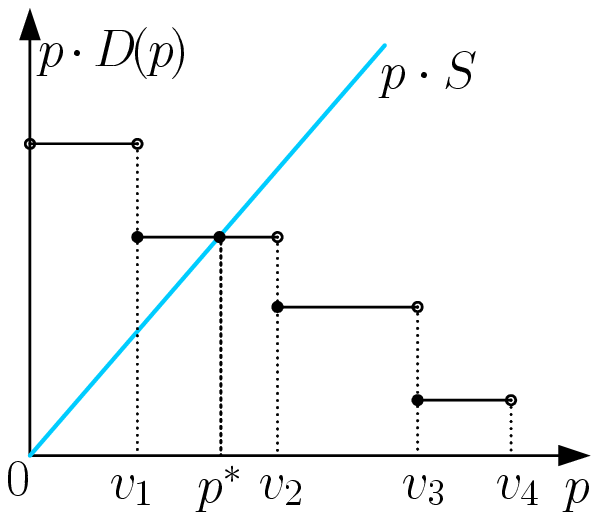}\\
    \end{minipage}
  }
  \subfloat[Undetermined Advertiser]{
    \label{SE_Profit_over_Price_2_ExPost}
    \begin{minipage}{0.4\textwidth}
        \centering
        \includegraphics[scale=0.75]{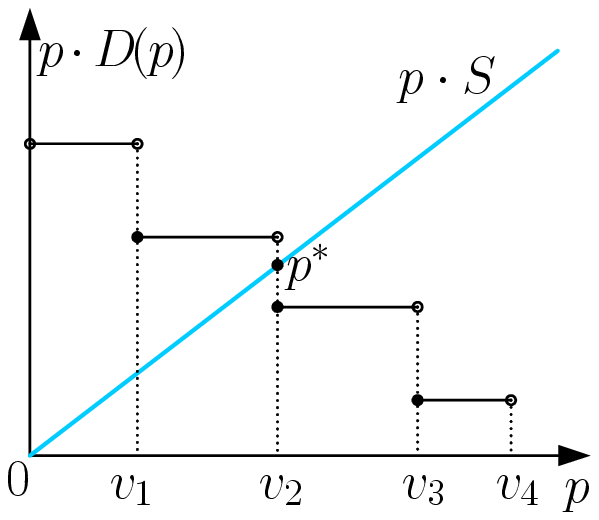}\\
    \end{minipage}
  }
  \caption{Search Engine Revenue Over Prices (Ex Post)}\label{fig:SE_Profit_over_Price_ExPost}
\end{figure}

We can now show a polynomial step algorithm for search engine to compute the optimal price. By
inputting the parameters of advertisers (assuming the indexes of advertisers are re-ordered such that
$v_i\leq v_{i+1}$, $i\in\{1,\ldots,m-1\}$), algorithm \ref{alg:calculate_opt_price} would return the
value of optimal price. The time complexity of the algorithm is $O(m^2)$ where $m$ is the number of
advertisers\footnote{The time complexity of the algorithm may be further reduced to $O(m)$ by
computing and saving the value of $sum_i\triangleq \sum_{j\in\{i,\ldots,m\}} B_j$,
$i\in\{1,\ldots,m\}$ first, which can be finished in $O(m)$ steps. Then the inner ``for'' loop in
algorithm \ref{alg:calculate_opt_price} can be substituted by the stored value of $sum_i$ and the
complexity of the algorithm is reduced to $O(m)$. Here for simplicity of exposition, we just show the
$O(m^2)$ algorithm.}.

\renewcommand{\algorithmicrequire}{\textbf{Begin}}
\renewcommand{\algorithmicensure}{\textbf{End}}
\begin{algorithm}
\caption{Calculate Optimal Price $p^*(\mathcal{I})$}\label{alg:calculate_opt_price}
\begin{algorithmic}[1]
    \REQUIRE
    \STATE $v_0=0$;
    \STATE \textbf{for} $i=1:m$
    \STATE $\qquad$$sum=0$;
    \STATE $\qquad$\textbf{for} $j=i:m$
    \STATE $\qquad$$\qquad$$sum+=B_j$;
    \STATE $\qquad$\textbf{end for};
    \STATE $\qquad$$p=sum/S$;
    \STATE $\qquad$\textbf{if} $(p\leq v_i)$
    \STATE $\qquad$$\qquad$\textbf{return} $\max(p,v_{i-1})$;
    \STATE $\qquad$\textbf{end if};
    \STATE \textbf{end for};
    \STATE \textbf{return} $v_m$;
    \ENSURE
\end{algorithmic}
\end{algorithm}

After determining the optimal price $p^*$, the quantity of attentions allocated to each advertiser
$i$, denoted by $q_i$, can be easily computed. The search engine would first find \emph{the least
index} of advertiser whose value is larger than or equal to $p^*$, which is denoted by
$j\in\{i\in\mathcal{I}:v_i\geq p^*, v_{i-1}<p^*\}$ (we define $v_0=0$). If there are no undetermined
advertisers, which implies $v_j>p^*$, the quantity allocated to advertiser $i$ would be $q_i=B_i/p^*$
for $i\geq j$ and $q_i=0$ otherwise. If undetermined advertiser does exist, which implies $v_j=p^*$,
we have $q_i=B_i/p^*$ for $i>j$, $q_j=S-\sum_{i=j+1}^m q_i$, and $q_i=0$ for $i<j$. In both cases,
the demand equals the supply, i.e., $\sum_{i\in\mathcal I}q_i=S$. This can be summarized in algorithm
\ref{alg:calculate_alloc_of_supply}.

\renewcommand{\algorithmicrequire}{\textbf{Begin}}
\renewcommand{\algorithmicensure}{\textbf{End}}
\begin{algorithm}
\caption{Calculate Allocation $q_i$, $i\in\mathcal{I}$}\label{alg:calculate_alloc_of_supply}
\begin{algorithmic}[1]
    \REQUIRE
    \STATE $sum=0$;
    \STATE $p=p^*$;
    \STATE \textbf{for} $j=1:m$
    \STATE $\qquad$\textbf{if} $(v_j\geq p)$
    \STATE $\qquad$$\qquad$\textbf{break;}
    \STATE $\qquad$\textbf{end if};
    \STATE \textbf{end for};
    \STATE \textbf{for} $i=1:(j-1)$
    \STATE $\qquad$$q_i=0;$
    \STATE \textbf{end for};
    \STATE \textbf{for} $i=(j+1):m$
    \STATE $\qquad$$q_i=B_i/p$;
    \STATE $\qquad$$sum+=q_i$;
    \STATE \textbf{end for};
    \STATE $q_j=S-sum$;
    \ENSURE
\end{algorithmic}
\end{algorithm}

The revenue of search engine is:
\begin{equation}
    R=p^*\cdot S
\end{equation}

The aggregate utility of advertisers is:
\begin{equation}
    U_A=\sum_{i\in\mathcal{I}^+(p^*)}(v_i-p^*)\frac{B_i}{p^*}
\end{equation}
Notice that the indifferent advertiser, if exists, would always achieve zero utility since the
current price equals its value, thus we don't need to consider it in the expression.

The \emph{social welfare} of the advertising system\footnote{We don't consider search users' utility
in the expression here.} is:
\begin{equation}\label{eq:defn:social_welfare}
    SW=R+U_A
\end{equation}

\begin{lem}\label{lem:price_over_ad_set}
The optimal price is non-decreasing over the set of participating advertisers given fixed supply $S$.
That is to say, for any advertisers set $\mathcal{I}_1$ and $\mathcal{I}_2$, if
$\mathcal{I}_1\subseteq \mathcal{I}_2$, we have $p^*(\mathcal{I}_1)\leq p^*(\mathcal{I}_2)$, where
$p^*(\mathcal{I})$ is obtained according to algorithm \ref{alg:calculate_opt_price}.
\end{lem}
\begin{proof}
We can prove the above lemma by contradiction. For simplicity of notation, we write $p_1\triangleq
p^*(\mathcal{I}_1)$ and $p_2\triangleq p^*(\mathcal{I}_2)$ and assume that $p_1>p_2$.

Since under optimal price, supply must be equal to demand, we have:
\begin{equation*}
    S=\sum_{i\in \mathcal{I}_1^+(p_1)}\frac{B_i}{p_1}+\frac{\alpha B_l}{p_1} \quad\mbox{and}\quad S=\sum_{i\in
    \mathcal{I}_2^+(p_2)}\frac{B_i}{p_2}+\frac{\beta B_{l'}}{p_2}
\end{equation*}
where $\alpha \in[0,1]$, and $\alpha>0$ if and only if there exists an indifferent advertiser $l$
whose value $v_l$ equals $p_1$; Similarly, $\beta\in[0,1]$, and $\beta>0$ if and only if there exists
an advertiser $l'$ such that $v_{l'}=p_2$.

For any advertiser $i\in \mathcal{I}_1^+(p_1)$, we have $i\in \mathcal{I}_1\subseteq \mathcal{I}_2$
and $v_i> p_1> p_2$, thus $i\in \mathcal{I}_2^+(p_2)$, which infers that
$\mathcal{I}_1^+(p_1)\subseteq \mathcal{I}_2^+(p_2)$; since $v_l=p_1>p_2$, we also have
$l\in\mathcal{I}_2^+(p_2)$. Therefore,
\begin{equation*}
    S=\sum_{i\in \mathcal{I}_1^+(p_1)}\frac{B_i}{p_1}+\frac{\alpha B_l}{p_1} \leq \sum_{i\in \mathcal{I}_1^+(p_1)}\frac{B_i}{p_1}+\frac{B_l}{p_1}\leq \sum_{i\in
    \mathcal{I}_2^+(p_2)}\frac{B_i}{p_1}<\sum_{i\in
    \mathcal{I}_2^+(p_2)}\frac{B_i}{p_2}\leq \sum_{i\in
    \mathcal{I}_2^+(p_2)}\frac{B_i}{p_2}+\frac{\beta B_{l'}}{p_2}
\end{equation*}
Contradiction to the conclusion that supply should equal demand under optimal price $p_2$.
\end{proof}

\begin{lem}\label{lem:revenue_over_ad_set}
The revenue of search engine is non-decreasing over the set of participating advertisers given fixed
supply $S$. That is to say, for any advertisers set $\mathcal{I}_1$ and $\mathcal{I}_2$, if
$\mathcal{I}_1\subseteq \mathcal{I}_2$, we have $R(\mathcal{I}_1)\leq R(\mathcal{I}_2)$.
\end{lem}
\begin{proof}
This conclusion can be deducted from lemma \ref{lem:price_over_ad_set} immediately:
\begin{equation*}
    R(\mathcal{I}_1)=p^*(\mathcal{I}_1)\cdot S \leq p^*(\mathcal{I}_2)\cdot S \leq R(\mathcal{I}_2).
\end{equation*}
\end{proof}

\begin{lem}\label{lem:price_over_supply}
The optimal price is non-increasing over the supply given the set of participating advertisers
$\mathcal{I}$. That is to say, for any supply $S_1, S_2\in[0,\infty)$, if $S_1>S_2$, we have
$p^*(S_1)\leq p^*(S_2)$.
\end{lem}
\begin{proof}
We prove the lemma by contradiction. For simplicity, we write $p_1\triangleq p^*(S_1)$ and
$p_2\triangleq p^*(S_2)$ and assume that $p_1>p_2$.

Since under optimal price, supply equals demand, we get:
\begin{equation*}
    S_1=\sum_{i\in \mathcal{I}^+(p_1)}\frac{B_i}{p_1}+\frac{\alpha B_l}{p_1} \quad\mbox{and}\quad S_2=\sum_{i\in
    \mathcal{I}^+(p_2)}\frac{B_i}{p_2}+\frac{\beta B_{l'}}{p_2}
\end{equation*}
where $\alpha \in[0,1]$, and $\alpha>0$ if and only if there exists an indifferent advertiser $l$
whose value $v_l$ equals $p_1$; Similarly, $\beta\in[0,1]$, and $\beta>0$ if and only if there exists
an advertiser $l'$ such that $v_{l'}=p_2$.

For any advertiser $i\in \mathcal{I}^+(p_1)$, we have $v_i> p_1> p_2$, thus $i\in
\mathcal{I}^+(p_2)$, which infers that $\mathcal{I}^+(p_1)\subseteq \mathcal{I}^+(p_2)$; since
$v_l=p_1>p_2$, we also have $l\in\mathcal{I}^+(p_2)$. Therefore,
\begin{equation*}
    S_1=\sum_{i\in \mathcal{I}^+(p_1)}\frac{B_i}{p_1}+\frac{\alpha B_l}{p_1} \leq \sum_{i\in \mathcal{I}^+(p_1)}\frac{B_i}{p_1}+\frac{B_l}{p_1}\leq \sum_{i\in
    \mathcal{I}^+(p_2)}\frac{B_i}{p_1}<\sum_{i\in
    \mathcal{I}^+(p_2)}\frac{B_i}{p_2}\leq \sum_{i\in
    \mathcal{I}^+(p_2)}\frac{B_i}{p_2}+\frac{\beta B_{l'}}{p_2}=S_2
\end{equation*}
Contradiction to our assumption of $S_1>S_2$.
\end{proof}

\begin{lem}
The revenue of search engine is non-decreasing over the supply given the set of participating
advertisers $\mathcal{I}$. That is to say, for any supply $S_1, S_2\in[0,\infty)$, if $S_1>S_2$, we
have $R(S_1)\geq R(S_2)$.
\end{lem}
\begin{proof}
For simplicity, we write $p_1\triangleq p^*(S_1)$ and $p_2\triangleq p^*(S_2)$ and from lemma
\ref{lem:price_over_supply} we know that $p_1<p_2$.

Since under optimal price, supply equals demand, we get:
\begin{equation*}
    S_1=\sum_{i\in \mathcal{I}^+(p_1)}\frac{B_i}{p_1}+\frac{\alpha B_l}{p_1} \quad\mbox{and}\quad S_2=\sum_{i\in
    \mathcal{I}^+(p_2)}\frac{B_i}{p_2}+\frac{\beta B_{l'}}{p_2}
\end{equation*}
where $\alpha \in[0,1]$, and $\alpha>0$ if and only if there exists an indifferent advertiser $l$
whose value $v_l$ equals $p_1$; Similarly, $\beta\in[0,1]$, and $\beta>0$ if and only if there exists
an advertiser $l'$ such that $v_{l'}=p_2$.

For any advertiser $i\in \mathcal{I}^+(p_2)$, we have $v_i> p_2> p_1$, thus $i\in
\mathcal{I}^+(p_1)$, which infers that $\mathcal{I}^+(p_2)\subseteq \mathcal{I}^+(p_1)$; since
$v_{l'}=p_2>p_1$, we also have $l'\in\mathcal{I}^+(p_1)$. Therefore,
\begin{equation*}
    R(S_2)=p_2\cdot S_2=\sum_{i\in \mathcal{I}^+(p_2)}B_i+\beta B_{l'}\leq \sum_{i\in
    \mathcal{I}^+(p_2)}B_i+B_{l'}\leq \sum_{i\in \mathcal{I}^+(p_1)}B_i \leq \sum_{i\in
    \mathcal{I}^+(p_1)}B_i+\alpha B_l=R(S_1)
\end{equation*}
\end{proof}

\subsection{Formulated As An Optimization Problem}
The revenue maximization problem confronting the monopolistic search engine could also be interpreted
as an optimization problem as follows:
\begin{eqnarray}
  \mbox{maximize} && p\cdot\sum_{i} q_i \label{eqn:max_mono}\\
  \mbox{subject to} && \forall i\in\mathcal{I}: p\cdot q_i\leq B_i \label{eqn:mono_cons_1}\\
   && \forall i\in\mathcal{I}: (v_i-p)\cdot q_i \geq 0 \label{eqn:mono_cons_2}\\
   && \sum_{i} q_i\leq S \label{eqn:mono_cons_3}\\
   && \forall i\in\mathcal{I}: p,q_i\geq 0 \label{eqn:mono_cons_4}
\end{eqnarray}
where the search engine needs to determine its optimal price $p$ and allocation of supply $q_i$ to
each advertiser $i$ in objective function \ref{eqn:max_mono}. Constraint (\ref{eqn:mono_cons_1})
means that each advertiser could not spend more than its budget. Constraint (\ref{eqn:mono_cons_2})
shows that the utility of advertiser must be non-negative, i.e., when the price $p$ exceeds the value
$v_i$, which implies that $q_i$ must be zero, advertiser $i$ would just quit the ad campaign. The
total supply to all advertisers is limited by $S$, which is shown in constraint
(\ref{eqn:mono_cons_3}). The last constraint states that all variables ($p$ and $q_i$, $i\in\mathcal
I$) should be non-negative.

As we have shown in proposition \ref{DemandEqualSupply} that the maximal revenue can only be obtained
when supply equals demand, therefore the above formulation may be further reduced to:
\begin{eqnarray}
  \mbox{maximize} && p\cdot S \label{eqn:max_mono_simplified}\\
  \mbox{subject to} && \forall i\in\mathcal{I}: p\cdot q_i\leq B_i \label{eqn:mono_cons_1_simplified}\\
   && \forall i\in\mathcal{I}: (v_i-p)\cdot q_i \geq 0 \label{eqn:mono_cons_2_simplified}\\
   && \sum_{i} q_i= S \label{eqn:mono_cons_3_simplified}\\
   && \forall i\in\mathcal{I}: p,q_i\geq 0 \label{eqn:mono_cons_4_simplified}
\end{eqnarray}
where constraint (\ref{eqn:mono_cons_3_simplified}) becomes tight and the objective function is
simplified to maximize variable $p$ only. Thus it's easy to see that the optimal solution for $p$ is
unique. Otherwise, if both $p^*$ and ${p^{*}}'$ maximize the objective function, it must be $p^*\cdot
S={p^{*}}'\cdot S$, so $p^*={p^{*}}'$. It remains to be inspected whether the optimal allocation
vector of $\vec{q}\triangleq (q_1,q_2,\ldots,q_m)$ for the optimization problem is unique too. We
summarize our conclusions in the following proposition.

\begin{prop}\label{prop:multi_alloc_vector}
In general, for revenue maximization problem (\ref{eqn:max_mono_simplified}), the optimal price $p$
is unique, however, there may be multiple optimal solutions for allocation vector $\vec{p}$.
\end{prop}
This can be shown by constructing a simple example as follows: there are two advertisers with
$v_1=1,v_2=2$ and $B_1=2,B_2=1$, search engine's supply is $S=2$. Assuming the optimal price $p^*$ is
larger than $1$, $q_1$ must be zero since $v_1<p^*$, then $q_2$ must be $2$ according to constraint
(\ref{eqn:mono_cons_3_simplified}). This would lead to $p\cdot q_2>2$, which contradicts with the
budget constraint of $B_2=1$. Now assuming the optimal price $p^*=1$, the constraints of the
optimization problem would be reduced to $q_1+q_2=2$ and $q_2\leq 1$, thus the optimal allocation
vector could be $\vec{q^*}=(q_1^*,q_2^*)=(2-q,q)$, $\forall q\in[0,1]$.

We now turn to investigate the effect of different optimal solutions to the social welfare. Since in
general $R=p\cdot \sum_{i\in \mathcal I} q_i$ and $U_A=\sum_{i\in \mathcal I} (v_i-p)q_i$, from
equation (\ref{eq:defn:social_welfare}) we get $SW=\sum_{i\in \mathcal I} v_i\cdot q_i$, where the
payments between search engine and advertiser are crossed off. If we examine the original
\emph{social welfare maximization} (SWM) problem under the same constraints
(\ref{eqn:mono_cons_1_simplified})-(\ref{eqn:mono_cons_4_simplified}), we can immediately present a
trivial solution as follows: $p=0$, $q_m=S$, $q_i=0$ for $i\in\mathcal{I}\backslash\{m\}$, i.e.,
letting the advertiser with the highest value acquire all the supply exclusively. Now the maximal
social welfare would be $SW_{max}=v_m\cdot S$. However, this solution is infeasible since it induces
zero profit to the search engine. An alternative problem the search engine may be interested in is to
maximize the social welfare while maintaining the optimal revenue it has achieved in
(\ref{eqn:max_mono_simplified}). In other words, we need to pick out one among the multiple optimal
allocations of (\ref{eqn:max_mono_simplified}) to maximize the social welfare. We call it as the
\emph{constrained social welfare maximization} (C-SWM) problem henceforth. This following theorem
gives the solution to the C-SWM problem.

\begin{theo}
Among all optimal solutions to the profit maximization problem (\ref{eqn:max_mono_simplified}),
algorithm \ref{alg:calculate_opt_price} and \ref{alg:calculate_alloc_of_supply} yield the one which
maximizes the social welfare, i.e., the solution to the constrained social welfare maximization
(C-SWM) problem.
\end{theo}
\begin{proof}
Denote the optimal price and allocation induced by algorithms as $p^*$ and $q_1^*,\ldots,q_m^*$, and
assuming there is another allocation vector $\hat{q}_1,\ldots,\hat{q}_m$ satisfying the constraints
(\ref{eqn:mono_cons_1_simplified})-(\ref{eqn:mono_cons_4_simplified}). Assuming advertiser $j$ has
the cutting-off value such that $v_i\geq p^*$ for $i\geq j$ and $v_i < p^*$ for $i<j$. Therefore from
constraint (\ref{eqn:mono_cons_2_simplified}), we must have $q_i^*=\hat{q}_i=0$ for all $i<j$.
According to algorithm \ref{alg:calculate_alloc_of_supply}, $q_i^*=\frac{B_i}{p^*}$ for all $i>j$ and
$q_j^*=S-\sum_{i=j+1}^m q_i$. Due to the constraint (\ref{eqn:mono_cons_2_simplified}), for all
$i>j$, it must be $\hat{q}_i\leq \frac{B_i}{p^*}=q_i^*$. Hence we can assume that
$\hat{q}_i=q_i^*-\delta_i$, where $\delta_i\geq 0$ for all $i>j$. Then the social welfare under
either allocation is:
\begin{eqnarray*}
  SW_1 &=& \sum_{i\in\mathcal I}v_i\cdot q_i^* \\ &=&\sum_{i=j+1}^m v_i\cdot q_i^* + v_j\cdot (S-\sum_{i=j+1}^m q_i^*) \\
  SW_2 &=& \sum_{i\in\mathcal I}v_i\cdot \hat{q}_i \\&=& \sum_{i=j+1}^m v_i\cdot (q_i^*-\delta_i) + v_j\cdot [S-\sum_{i=j+1}^m
  (q_i^*-\delta_i)] \\ &=& SW_1+ \sum_{i=j+1}^m \delta_i\cdot (v_j-v_i)
\end{eqnarray*}
Since $\delta_i\geq 0$ and $v_j\leq v_i$ for $i>j$, we have $SW_2\leq SW_1$. Therefore
$q_1^*,\ldots,q_m^*$ maximizes the social welfare over all possible optimal allocations of
(\ref{eqn:max_mono_simplified}).
\end{proof}

\section{The Duopoly Market Model}\label{sec:4:duo_model}
In this section we switch from the monopoly model to the competitive model with more than one search
engine. Considering the likely situation where there is one leading search company and one major
competitor in the market (for example, Google and Yahoo! in the United States), we describe a duopoly
model where one search engine has an advantage over the other. We formulate their competition as a
three-stage dynamic game and solve it from the \emph{ex post} perspective as follows.

\subsection{Competition for End Users in Stage I}
In Stage I search engines would choose different strategies for attracting end users with different
tastes. The user bases they attract in this Stage would be the decisive factor for determining their
supply of user attentions to advertisers in subsequent stages.

We assume that there are two \emph{horizontally} and \emph{vertically} differentiated search engines
$\mathcal{J}=\{1,2\}$ providing search results to users and selling ad opportunity to advertisers.

Here \emph{horizontal difference} means the different design of their home pages and diversity of
extra services such as email, news and other applications. Different users may have different tastes
and preferences and hence be attracted by different search engines.

\emph{Vertical difference} means the quality of searching results. The higher the quality is, the
better users and advertisers would feel. We assume that search engine 1 possesses the leading
technology to match ads to search queries and can provide better service for both users and
advertisers than search engine 2.

In terms of horizontal difference, the canonical Hotelling's model of spatial competition
\cite{Hotelling} provides an appealing framework to address the equilibrium in characteristic space.
The behaviors of providers could then be rationalized as the best-response strategies of players in a
location game. The dynamics of the game can be described as follows: each provider chooses a location
in the characteristic space which denotes the specific feature of service it provides to users. And
each user is characterized by an address reflecting his individual preference of ideal features
search engines should provide. Searching at engine $j\in\mathcal J$ involves quadratic transportation
cost\footnote{Actually in the seminar paper of Hotelling \cite{Hotelling} the author assumed the
linear transportation cost, which resulted in no equilibrium results. Later literatures on
Hotelling's model usually modified this assumption to the quadratic transportation cost which ensures
existence of equilibrium. Here we followed this line of revised model as applied in recent papers
such as \cite{Location, circle_model}. Interesting readers may further refer to the excellent survey
of \cite{Diff_survey} for a comprehensive discussion and review of different variants of the
Hotelling's model.} for a user if engine $j$ is not located in his ideal position. Users would choose
search engine which provides better search results and also induces as low transportation cost as
possible.

Assuming that a continuum of users are spread uniformly with unit density on the circumference of a
unit circle. The address of user is denoted by $t\in[0,1)$. Without loss of generality, let search
engine 1 locate at $x_1=0$\footnote{Since it is a circle, engine 1 can also be regarded as locating
at the ending point $x_1=1$.} and search engine 2 $x_2\in[0,1)$, as shown in figure
\ref{fig:Circular}.

\begin{figure}[ht]
  \centering
  \includegraphics[scale=1]{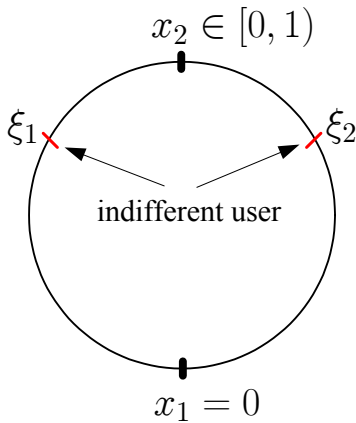}\\
  \caption{Users in Circular Domain}\label{fig:Circular}
\end{figure}

Assuming search engine 1 can provide higher quality results for users than search engine 2. Then the
utility of the user searching in either engine would be as follows:
  \begin{eqnarray}
    u_1(t) &=& \zeta_1 q-C(t,x_1)=q-\min\{t^2,(1-t)^2\} \\
    u_2(t) &=& \zeta_2 q-C(t,x_2)=\zeta q-(t-x_2)^2
  \end{eqnarray}
where $\zeta\in[0,1]$ denotes the comparative ``disability'' of search engine 2 to provide the best
search result to users; $q$ is the positive payoff users perceive when certain information is
returned by the search engine for a particular query; $C(t, x_j)$ is the transportation cost incurred
when there is some distance between user's address $t$ and search engine $j$'s location $x_j$.

Let $u_1(\xi)=u_2(\xi)$ we can find the location of users who are indifferent between searching in
two engines:
\begin{eqnarray*}
  \xi_1 &=& \frac{(1-\zeta)q+x_2^2}{2x_2} \\
  \xi_2 &=& \frac{1-x_2^2-(1-\zeta)q}{2(1-x_2)}
\end{eqnarray*}

Then the market share of search engine 2 is
\begin{equation*}
    n_2(x_2)=\xi_2-\xi_1=\frac{1}{2}[1-\frac{(1-\zeta)q}{x_2(1-x_2)}]
\end{equation*}
and search engine 1 obtains the remaining market share: $n_1=1-n_2$. By applying the first-order
condition $\frac{d n_2}{d x_2}=0$, we have $x_2^*=\frac{1}{2}$, i.e., the maximum differentiation.

Letting $x_2=\frac{1}{2}$ we have
\begin{eqnarray*}
    n_1 &=&  \frac{1}{2}+2(1-\zeta)q \\
    n_2 &=& \frac{1}{2}- 2(1-\zeta)q
\end{eqnarray*}

As we can see, when two search engines provide the same quality of service ($\zeta=1$), they will
divide the market share equally. The less quality search engine 2 provides, the less market share it
can hold.

Since the impression number for a particular keyword in a search engine is proportional to the users
it attracts: the more users see the advertisement, the more impressions the ad would receive in
general. To be aligned with the monopoly case in previous section, here we assume the total supply is
still $S$ and the supply of each search engine is denoted by:
\begin{eqnarray*}
  S_1=S\cdot\frac{n_1}{n_1+n_2}=S\cdot n_1 \\
  S_2=S\cdot\frac{n_2}{n_1+n_2}=S\cdot n_2
\end{eqnarray*}
Since $n_1\geq n_2$, we have also $S_1\geq S_2$.

\subsection{Competition for Advertisers in Stage II and III}
Search engines compete for advertisers in the last two stages to maximize their revenues
\emph{subject to} the supply constraint $(S_1, S_2)$ determined in Stage I. In Stage II, search
engines determine their optimal prices $(p_1,p_2)$ for charging advertisers; and consequently in
Stage III, advertisers would choose their favorite search engine for advertisements based on the
previously announced prices. Facing the new advertiser sets in Stage III, search engines may want to
revert to the second stage and revise their optimal prices, and consequently, advertisers would make
necessary adjustment in the third stage. Therefore, Stage II and III would \emph{alternate}
dynamically until it reaches certain stable state. we will discuss this dynamic process in details in
the following section.

For advertiser $i\in \mathcal{I}$, the utility of participating in the ad campaign in either search
engine is:
\begin{eqnarray}
  \pi_1^i &=& \max \{(v_i-p_1)\frac{B_i}{p_1}\, , \,0\} \label{eq:uti:adv} \\
  \pi_2^i &=& \max \{(v_i\rho_i-p_2)\frac{B_i}{p_2}\, ,\, 0\}
\end{eqnarray}
where $\rho_i\in[0,1]$ is called \emph{discount factor} denoting advertiser $i$'s perceived
``disability'' of search engine 2 to convert the impressions to clicks (or sales of products). We
assume that search engine 1 owns better technology and is able to match users' interest with the most
suitable ads, hence can generate a higher \emph{click-through rate} (users' probability of clicking
after seeing the ads) or \emph{conversion rate} (users' probability of purchase the product or
service after clicking the ads) than search engine 2. So in general advertisers would evaluate each
impression in search engine 1 higher than in engine 2. For simplicity of notation, we have normalized
the discount factor of per-impression value in search engine 1 as unity. In practical market,
advertisers can be roughly classified into two categories: \emph{brand advertisers} and
\emph{performance advertisers} \cite{branding_adv}. Brand advertisers usually have higher $\rho$
since they aim to promote the brand awareness among users and hence the relative technology
disadvantage in search engine 2 would have less effect on their values for each attention/impression.
However, for performance advertisers who care more on the click-through rate or conversion rate, the
technology disadvantage would affect their values for each impression more and therefore result in
lower values of $\rho$. To be more exact, we let the expectation $E(\rho)$ of discount factor serve
as the cutting-off value for two types of advertisers, i.e., advertisers with higher $\rho$ than
$E(\rho)$ is defined as brand advertisers and the others are performance advertisers in our model.

By letting $\pi_1^i\geq \pi_2^i$ we can derive the condition under which advertiser $i$ would choose
search engine 1:
\begin{equation*}
    \rho_i\leq \frac{p_2}{p_1}
\end{equation*}

Assuming that advertisers are re-ordered according to $\rho_i$. Then the division of advertisers can
be depicted in figure \ref{fig:div_advertisers} where $\mathcal{I}_1(p_1,p_2)=\{i\in
\mathcal{I}:\rho_i\leq \frac{p_2}{p_1}\}$ denotes the set of advertisers who prefer search engine 1
and $\mathcal{I}_2(p_1,p_2)=\{i\in \mathcal{I}:\rho_i> \frac{p_2}{p_1}\}$ the set of advertisers
preferring engine 2.

\begin{figure}[ht]
  \centering
  \includegraphics[scale=1]{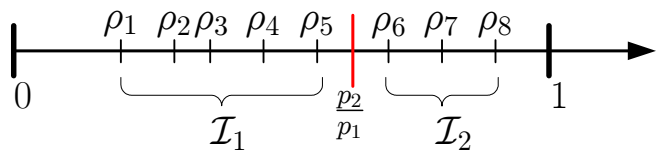}\\
  \caption{The Division of Advertisers}\label{fig:div_advertisers}
\end{figure}

After initial price $p_1$ and $p_2$ are set in the market, the advertisers set is divided into
$\mathcal{I}_1$ and $\mathcal{I}_2$. Then each search engine can compute its optimal price
${p}^*_1(\mathcal{I}_1)$ and ${p}^*_2(\mathcal{I}_2)$ independently as the monopoly case and price
ratio $\frac{{p}^*_2(\mathcal{I}_2)}{{p}^*_1(\mathcal{I}_1)}$ gets updated. If it happens that the
new price ratio divides the advertisers set into $\mathcal{I}_1$ and $\mathcal{I}_2$, we say this is
a \emph{Nash equilibrium (NE) price pair} as $(p_1^{NE}, p_2^{NE})$ and neither search engine has
incentive to deviate unilaterally. Otherwise, the process will iterate until the prices become
stable.

Defining first the set of advertisers who participate the advertising campaign as follows.
\begin{eqnarray}
  \mathcal{I}_1^+(p_1,p_2) &\triangleq& \{i\in \mathcal{I}: \rho_i\leq\frac{p_2}{p_1}, v_i\geq p_1\}\label{eq:def_I_1} \\
  \mathcal{I}_2^+(p_1,p_2) &\triangleq& \{i\in \mathcal{I}: \rho_i>\frac{p_2}{p_1}, \rho_i v_i\geq
  p_2\}\label{eq:def_I_2}
\end{eqnarray}

We now give the formal definition of NE price pair.
\begin{defn}
A price pair of $(p_1,p_2)$ is called a \emph{Nash equilibrium price pair} if
$p_1=p^*(\mathcal{I}_1^+(p_1,p_2))$ and $p_2=p^*(\mathcal{I}_2^+(p_1,p_2))$ where $p^*(\mathcal{I})$
is computed according to algorithm \ref{alg:calculate_opt_price}.
\end{defn}

It's easy to see that under the NE price pair $(p_1,p_2)$, for any advertiser $i\in\mathcal{I}_1^+$
or $i\in\mathcal{I}_2^+$, it would have no incentive to switch to the other search engine; for
advertiser $i\in \mathcal{I}_1 \backslash \mathcal{I}_1^+$, since $\rho_i\leq\frac{p_2}{p_1}$ and
$v_i<p_1$, it holds that $\rho_i v_i<p_2$, thus $i$ would not switch to engine 2 which generates zero
utility according to equation (\ref{eq:uti:adv}); for advertiser $i\in \mathcal{I}_2 \backslash
\mathcal{I}_2^+$, since $\rho_i>\frac{p_2}{p_1}$ and $\rho_i v_i<p_2$, we have $v_i<p_1$ so $i$ have
no incentive to switch to engine 1. Thus the NE price pair would induce a stable state to the
competition system.

\begin{prop}\label{prop:NE_maynot_exist}
None-zero NE price pairs may not exist.
\end{prop}
A simple counter-example to illustrate proposition \ref{prop:NE_maynot_exist} is when there is only
one advertiser in the system. No matter which search engine this advertiser chooses, the price in
\emph{the other} search engine would be zero since it attracts no advertisers. Then the advertiser
would have incentive to join \emph{the other} search engine due to the zero price. However, once the
advertiser switches, the price in the other search engine would become positive and price in the
original engine decreases to zero. Thus the advertiser would keep switching between two search
engines and no stable prices can be reached.

\begin{prop}\label{prop:price_relation}
If NE price pair $(p_1^{NE}, p_2^{NE})$ exists, it must be $p_1^{NE}\geq p_2^{NE}$.
\end{prop}
\begin{proof}
Assuming $p_1^{NE}< p_2^{NE}$, then since $\frac{p_2^{NE}}{p_1^{NE}}>1$, we have
$\mathcal{I}_2^+(p_1,p_2)=\varnothing$. Therefore $p_2^{NE}=0$ and $p_1^{NE}< 0$. However, it cannot
be the case since rational search engine would never set negative prices.
\end{proof}

\begin{prop}
In the stable state, search engine 2 cannot make higher revenue than engine 1.
\end{prop}
\begin{proof}
Since $R_1=p_1^{NE}\cdot S_1$ and $R_2=p_2^{NE}\cdot S_2$, from proposition \ref{prop:price_relation}
and $S_1\geq S_2$, it's easy to see that $R_1\geq R_2$.
\end{proof}

Denote $\nu$ as the price ratio $\frac{p_2}{p_1}$ which determines advertisers' preferences, we
define the optimal price ratio as:
\begin{equation*}
    f(\nu)=\frac{p_2^*(\mathcal{I}_2(\nu))}{p_1^*(\mathcal{I}_1(\nu))}
\end{equation*}
where $p^*(\mathcal{I})$ is obtained according to algorithm \ref{alg:calculate_opt_price}. If the
advertiser partitions generated by $\nu$ are the same as those generated by the optimal price ratio
$f(\nu)$, then the partitions are stable and the optimal prices become NE price pair. The problem
reduces to find the \emph{fixed points} which satisfy $f(\nu^*)=\nu^*$. Notice that $f(\nu)$ is
piece-wise constant and its value changes at $\nu=\rho_i$, $i\in\{1,\ldots,m\}$.

From the definitions of $\mathcal{I}_1$ and $\mathcal{I}_2$ we see that as $\nu$ increases, the
preferred set of $\mathcal{I}_1$ would expand while $\mathcal{I}_2$ shrinks. According to lemma
\ref{lem:price_over_ad_set}, we know that $p_1^*$ would increase while $p_2^*$ decreases. Therefore
$f(\nu)$ should be a non-increasing function of $\nu$.

We can now show the dynamics of function $f(\nu)$ in figure \ref{Price_ratio}. In this example we
assume there are five advertisers re-ordered by their values of $\rho$ such that $\rho_i\leq
\rho_{i+1}$, $i\in\{1,2,3,4\}$. Similarly, there may be two different scenarios for the location of
fixed point $v^*$: (a) $v^*\in(\rho_i,\rho_{i+1})$, $i\in\{1,\ldots,m-1\}$ as shown in figure
\ref{Price_ratio_1}, and (b) $v^*=\rho_i$, $i\in\{1,\ldots,m\}$ as shown in figure
\ref{Price_ratio_2}.

\begin{figure}[ht]
  \centering
  \subfloat[Determined Division]{
    \label{Price_ratio_1}
    \begin{minipage}{0.4\textwidth}
        \centering
        \includegraphics[scale=0.75]{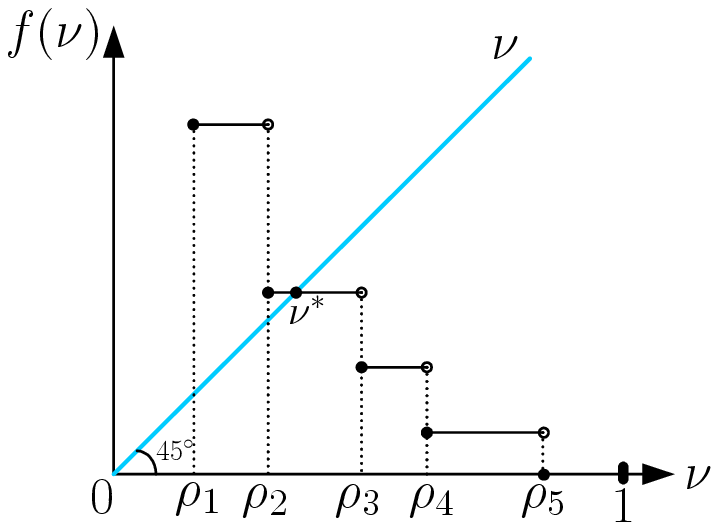}\\
    \end{minipage}
  }
  \subfloat[Undetermined Division]{
    \label{Price_ratio_2}
    \begin{minipage}{0.4\textwidth}
        \centering
        \includegraphics[scale=0.75]{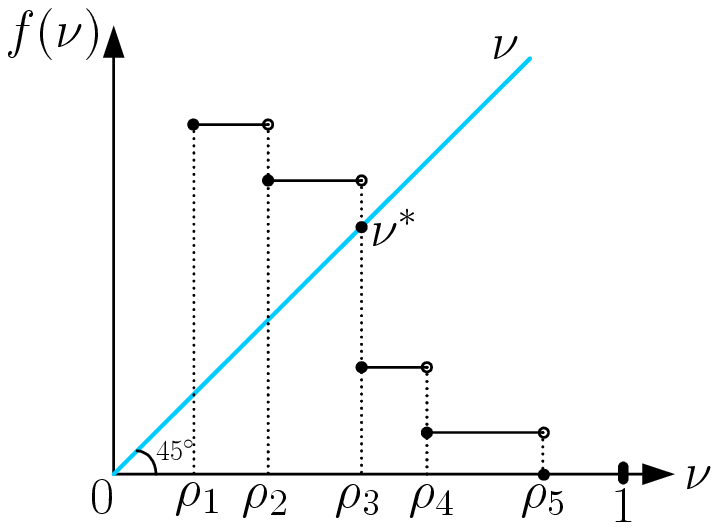}\\
    \end{minipage}
  }
  \caption{Division of Advertisers Set by their Preferences}\label{Price_ratio}
\end{figure}

For case (a), the optimal $\nu^*$ divides the advertisers set into exactly two subsets, and all
advertisers with strict preferences over certain search engine are aggregated to one of the subsets.
The market would become stable after each search engine sets their optimal price. For case (b),
however, there is one special advertiser who would keep switching from one search engine to the
other. To illustrate it, assuming the index of this special advertiser is $l$ which satisfies the
following condition:
\begin{equation*}
    \frac{p_2^*(\{l,l+1,\ldots,m\})}{p_1^*(\{1,2,\ldots,l-1\})}>\rho_l>\frac{p_2^*(\{l+1,\ldots,m\})}{p_1^*(\{1,2,\ldots,l\})}
\end{equation*}
The first inequality above implies that advertiser $l$ prefers search engine 1 if $l$ has already
joined the system of search engine 2, while the second inequality implies that he would prefer engine
2 if he is associated with engine 1. Therefore, advertiser $l$ would keep switching between two
search engines. In this case, we call advertiser $l$ as the \emph{undetermined} advertiser.

The undetermined advertiser problem arises from our assumption that advertisers can only purchase
service from one engine. This resembles the classic \emph{oscillation} problem in the multi-path
routing when all traffic is aggregated to the least congested path \cite{Multipath_routing}. Similar
to the ``splittable'' model in \cite{Selfish_routing} where network users are permitted to route
traffic fractionally over many paths, we can also make the following assumption for our model:

\begin{assumption}
\emph{Splittable Budgets:} We assume advertisers can arbitrarily split their budgets and invest them
into both search engines to maximize their utility.
\end{assumption}

Under this splittable assumption, the dynamics of undetermined advertiser $l$'s strategic behavior
could be interpreted as follows: assuming starting from the initial state where advertiser $l$ has
joined engine 1, and is facing a lower price ratio ($\frac{p_2}{p_1}<\rho_l$) which indicates him to
invest more on engine 2; then $l$ would try to split his budget into two parts: $(1-\alpha) B_l$ goes
to engine 1 and the rest of $\alpha B_l$ goes to engine 2, with $\alpha\in[0,1]$. As advertiser $l$
invests more and more budgets on engine 2, i.e., $\alpha$ keeps growing from zero to one, the price
ratio $\frac{p_2}{p_1}$ would keep rising until at certain $\alpha^*\in(0,1)$ it equals $\rho_l$ and
advertiser $l$ would have no incentive to invest more on engine 2.

It remains to be shown whether the price ratio above would increase ``smoothly''\footnote{To be more
exact, we need to guarantee that there are no discontinuous points. Otherwise, the optimal $\alpha^*$
may not exist.} as $\alpha$ increases and whether there always exists $\alpha^*$ for the undetermined
advertiser to divide his budget. We summarize our conclusion in the following theorem.
\begin{theo}
Assuming that there exists an undetermined advertiser $l$ and this advertiser can purchase service
from both search engines. In particular, advertiser $l$ can arbitrarily split his budget into
$(1-\alpha) B_l$ and $\alpha B_l$ with $\alpha\in[0,1]$, where the former is invested to engine 1 and
the latter to engine 2. Then there must exist $p_1^*$, $p_2^*$ and $\alpha^*\in[0,1]$ such that
$\frac{p_2^*}{p_1^*}=\rho_l$.
\end{theo}
\begin{proof}
Denote $\mathcal{I}_1\triangleq \{1,2,\ldots,l-1\}$ and $\mathcal{I}_2\triangleq
\{l+1,l+2,\ldots,m\}$. By previous analysis we know that advertisers in $\mathcal{I}_1$ would always
be associated with engine 1 and advertisers in $\mathcal{I}_2$ be associated with engine 2. It
remains to be shown the effect of advertiser $l$'s splitting decision on the price ratio. Define
first that $p_1^0=p_1^*(\mathcal{I}_1)$, $p_1^1=p_1^*(\mathcal{I}_1\cup \{l\})$ and
$p_1^\alpha=p_1^*(\mathcal{I}_1\cup \alpha\{l\})$ where $\alpha\{l\}$ denotes that advertiser $l$
participates in search engine 1, but with fractional budget of $\alpha B_i$; similarly, we define
$p_2^0=p_2^*(\mathcal{I}_2)$, $p_2^1=p_2^*(\mathcal{I}_2\cup \{l\})$ and
$p_2^\alpha=p_2^*(\mathcal{I}_2\cup \alpha\{l\})$.

To prove that $\frac{p_2^\alpha}{p_1^{1-\alpha}}$ changes \emph{continuously} from
$\frac{p_2^0}{p_1^1}$ to $\frac{p_2^1}{p_1^0}$ as $\alpha$ increases from zero to one, we first
present lemma \ref{prop:continuity_as_budget} as follows.

\begin{lem}\label{prop:continuity_as_budget}
Given a monopolistic search engine and the advertisers set $\mathcal{I}$. Assuming all parameters of
the ad system, including engine's supply, advertisers' values and budgets, are definite except the
budget $B_i$ of certain advertiser $i\in\mathcal{I}$, then the optimal price $p^*(\mathcal{I})$ can
be regarded as a function over the variable $B_i\geq 0$. Furthermore, the function is continuous and
non-decreasing as $B_i$ increases.
\end{lem}
\begin{proof}
Since from algorithm \ref{alg:calculate_opt_price}, for each realization of $B_i$, we can compute a
definite value of optimal price, the optimal price can therefore be regarded as a function of $B_i$.

The property of non-decreasing (or weakly increasing) is easy to see by contradiction. Assuming for
$B_i$ the optimal price is $p_1$ and for $B_i+\delta$ ($\delta>0$) the optimal price is $p_2$ with
$p_2<p_1$. Then $p_1$ would be the optimal solution for formulation in
(\ref{eqn:max_mono_simplified})-(\ref{eqn:mono_cons_4_simplified}) and we let $\vec{q^*}$ be one of
the binding optimal allocations. After we increase advertiser $i$'s budget to $B_i+\delta$, the
previous solution combination $(p_1, \vec{q^*})$ would still satisfy the constraints of the revised
optimization problem since $p_1\cdot q_i^*\leq B_i\leq B_i+\delta$ and all other conditions remain
unchanged. Since $p_2$ is the optimal solution for the revised problem which maximizes the objective
function of $p\cdot S$, we have $p_2\cdot S\geq p_1\cdot S$, thus, $p_2\geq p_1$. This contradicts
with our previous assumption of $p_2<p_1$.

We now turn to prove the property of continuity. Let $p^*(B_i)$ denote the optimal price under budget
$B_i$ and $p^*(B_i+\varepsilon)$ the optimal price under budget $B_i+\varepsilon$ where $\varepsilon$
is any small real number. As shown in figure \ref{fig:SE_Profit_over_Price_ExPost}, for arbitrary
budget $B_i\in[0,\infty)$, there exist two different scenarios for computing the optimal price: (a)
all advertisers are determined; (b) there is one undetermined advertiser. We will discuss these two
cases separately as follows.

Case (a): $v_l<p^*<v_{l+1}$, $l\in\{0,1,\ldots,m-1\}$ (let $v_0=0$). We further consider two cases
for the index of advertiser $i$ whose budget $B_i$ is the variable: (i) $i\leq l$; (ii) $i>l$. For
type (i), since $p^*(B_i)=(B_{l+1}+B_{l+2}+\cdots+B_m)/S$, the change of $B_i$ would not affect the
value of optimal price, therefore, we have $\lim\limits_{\varepsilon\to 0}
p^*(B_i+\varepsilon)=p^*(B_i)$; For type (ii), since advertiser $i$ is in the participating set, the
change of $B_i$ does affect the optimal price. Assuming $\varepsilon$ is small enough such that
$|\varepsilon|<\min\{(p^*-v_l)\cdot S, (v_{l+1}-p^*)\cdot S\}$. This condition guarantees that
$p^*(B_i+\varepsilon)=p^*(B_i)+\varepsilon/S$ is still in the interval of $(v_l,v_{l+1})$. Therefore,
$\lim\limits_{\varepsilon\to 0} p^*(B_i+\varepsilon)=p^*(B_i)+\lim\limits_{\varepsilon\to
0}\varepsilon/S=p^*(B_i)$.

Case (b): $p^*=v_l$, $l\in\{1,\ldots,m\}$ and $B_{l+1}+\cdots+B_m\leq v_l\cdot S\leq
B_{l}+B_{l+1}+\cdots+B_m$ where advertiser $l$ would only consume part of his budget under price
$v_l$. For $i< l$, the change of $B_i$ would not affect $p^*$, so we only need to consider the case
for $i\geq l$. We now consider three possible scenarios for $v_l\cdot S$:
\begin{enumerate}[(i)]
  \item $B_{l+1}+\cdots+B_m<v_l\cdot S<B_{l}+B_{l+1}+\cdots+B_m$. Assuming $\varepsilon$ is small enough such
  that $|\varepsilon|<\min\{v_l\cdot S-(B_{l+1}+\cdots+B_m), B_{l}+B_{l+1}+\cdots+B_m-v_l\cdot S\}$.
  This condition guarantees that after $B_i$ has changed $\varepsilon$, $v_l\cdot S$ is still in the
  interval of $(B_{l+1}+\cdots+B_m+\varepsilon,B_{l}+B_{l+1}+\cdots+B_m+\varepsilon)$, which means
  that $\lim\limits_{\varepsilon\to 0} p^*(B_i+\varepsilon)=\lim\limits_{\varepsilon\to 0}
  v_l=v_l=p^*(B_i)$;
  \item $v_l\cdot S=B_{l+1}+\cdots+B_m$. When $\varepsilon$ is negative, it will be equivalent to the
  above case (i) and therefore we have $\lim\limits_{\varepsilon\to 0^-}
  p^*(B_i+\varepsilon)=\lim\limits_{\varepsilon\to 0^-} v_l=p^*(B_i)$; when $\varepsilon$ is
  positive, it will be equivalent to case (a) and therefore $\lim\limits_{\varepsilon\to 0^+}
  p^*(B_i+\varepsilon)=p^*(B_i)+\lim\limits_{\varepsilon\to 0^+} \varepsilon/S=p^*(B_i)$;
  \item $v_l\cdot S=B_{l}+B_{l+1}+\cdots+B_m$. When $\varepsilon$ is negative, it will be equivalent
  to case (a) and therefore $\lim\limits_{\varepsilon\to 0^-}
  p^*(B_i+\varepsilon)=p^*(B_i)+\lim\limits_{\varepsilon\to 0^-} \varepsilon/S=p^*(B_i)$; when $\varepsilon$ is
  positive, it will be equivalent to  the
  above case (i) and therefore we have $\lim\limits_{\varepsilon\to 0^+}
  p^*(B_i+\varepsilon)=\lim\limits_{\varepsilon\to 0^+} v_l=p^*(B_i)$.
\end{enumerate}

Therefore now we can conclude that for any $B_i\in[0,\infty)$, it always holds that
$\lim\limits_{\varepsilon\to 0} p^*(B_i+\varepsilon)=p^*(B_i)$. So $p^*(B_i)$ is a continuous
function over $B_i\in[0,\infty)$.
\end{proof}

From lemma \ref{prop:continuity_as_budget} we see that $p_2^\alpha$ is continuous and weakly
increasing function of $\alpha$ and $p_1^{1-\alpha}$ is continuous and weakly decreasing function of
$\alpha$, thus the price ratio $\frac{p_2^\alpha}{p_1^{1-\alpha}}$ is continuous and weakly
increasing from $\frac{p_2^0}{p_1^1}$ to $\frac{p_2^1}{p_1^0}$ as $\alpha$ rises from zero to one.
Thus there must exist an optimal $\alpha^*\in[0,1]$ such that
$\frac{p_2^{\alpha^*}}{p_1^{1-\alpha^*}}=\rho_l\in(\frac{p_2^0}{p_1^1},\frac{p_2^1}{p_1^0})$.
\end{proof}

\subsection{Comparison of Competition and Monopoly}
After showing the existence of Nash equilibrium prices under a relaxed assumption, we can apply this
NE outcome to predict the revenue and social welfare in the duopoly environment, and compare them
with the corresponding results when one search engine monopolize the market. These comparative
results would be instructive in practice considering the attempt of cooperation among large search
companies such as Google and Yahoo!.\footnote{In June 2008, Google and Yahoo! announced an
advertising cooperation agreement which was later on forced to be abandoned due to antitrust concern
of government regulators. See the article ``Antitrust Concerns Kill Yahoo-Google Ad Deal,'' CNET,
November 5, 2008 (\url{http://news.cnet.com/8301-1023_3-10082800-93.html}).}

We now turn to compare the prices under competition and monopoly. The main results are given in the
following theorem.

\begin{theo}\label{theo:price_duo_vs_price_mono}
The equilibrium price in search engine 1 (or engine 2) under competition is no less (or larger) than
the optimal price when engine 1 monopolizes the market.
\end{theo}

\begin{proof}
Assuming all discount factors are randomly drawn in the range of
$(\overline{\rho},\underline{\rho})$. According to equation (\ref{eq:S_equal_D}) and proposition
\ref{DemandEqualSupply}, the monopoly price $p^*$ satisfies the following condition:

\begin{equation}\label{eq:0}
    p^*\cdot S=m\cdot E(B)\cdot [1-F(p^*)]
\end{equation}

Now we divide the total supply arbitrarily into $S_1=\alpha S$ and $S_2=(1-\alpha)S$,
$\alpha\in[0,1]$ for search engine 1 and 2. Suppose the optimal prices are $p_1$ and $p_2$
respectively and there are $m_1$ advertisers attracted by engine 1 and $m_2$ advertisers by engine 2
where $m_1+m_2=m$, by applying equation (\ref{eq:S_equal_D}) and proposition \ref{DemandEqualSupply}
we get the following equations for both search engines,

\begin{eqnarray}
    p_1\cdot \alpha S &=& m_1\cdot {E(B)}\cdot[1-F(p_1)] \label{eq:1}  \\
    p_2\cdot (1-\alpha) S &=& \sum_{i\in\mathcal{I}_2} {E(B)}\cdot\mbox{Prob}\{\rho_i v_i>p_2\} \label{eq:2}
\end{eqnarray}

since in equilibrium $p_2/p_1=\nu^*$ and for each advertiser $i\in\mathcal{I}_2$ it holds that
$\rho_i\geq \nu^*$, we can derive that:
\begin{equation*}
    \mbox{Prob}\{\rho_i v_i>p_2\}=\mbox{Prob}\{\rho_i v_i>\nu^* p_1\}\geq\mbox{Prob}\{v_i>p_1\}=1-F(p_1)
\end{equation*}

Now the equation (\ref{eq:2}) would become:

\begin{equation}\label{eq:3}
    \nu^*(1-\alpha)\cdot p_1 S\geq m_2\cdot {E(B)}\cdot[1-F(p_1)]
\end{equation}

Summing over conditions (\ref{eq:1}) and (\ref{eq:3}), we have
\begin{equation}\label{eq:4}
    [\alpha+\nu^*(1-\alpha)]\cdot p_1 S\geq m\cdot E(B)\cdot [1-F(p_1)]
\end{equation}

Since $\nu^*\leq 1$, we have $\alpha+\nu^*(1-\alpha)\leq \alpha+(1-\alpha)=1$. Defining the following
function first,
\begin{equation*}
    h(p)\triangleq \frac{1-F(p)}{p}
\end{equation*}
which is strictly decreasing over $p$. Then by comparing conditions (\ref{eq:0}) and (\ref{eq:4}), we
get
\begin{equation*}
    h(p_1)\leq \frac{[\alpha+\nu^*(1-\alpha)]\cdot S}{m\cdot E(B)}\leq \frac{S}{m\cdot E(B)}=h(p^*)
\end{equation*}
we can infer that $p_1\geq p^*$.

Since from proposition \ref{prop:price_relation} we have $p_1\geq p_2$ and $1-F(p)$ is a monotonic
decreasing function of $p$, equation (\ref{eq:1}) would become:
\begin{equation}\label{eq:5}
    p_2\cdot \alpha S   \leq p_1\cdot \alpha S  \leq   m_1\cdot {E(B)}\cdot[1-F(p_2)]
\end{equation}

And since $\rho_i\leq 1$ for any $i\in\mathcal{I}_2$, we know that:
\begin{equation*}
    \mbox{Prob}\{\rho_i v_i>p_2\}\leq   \mbox{Prob}\{v_i>p_2\}=1-F(p_2)
\end{equation*}
thus equation (\ref{eq:2}) would become:
\begin{equation}\label{eq:6}
    p_2\cdot (1-\alpha) S \leq m_2\cdot {E(B)}\cdot[1-F(p_2)]
\end{equation}
Summing over inequalities (\ref{eq:5}) and (\ref{eq:6}), we get
\begin{equation*}
    p_2 S \leq m\cdot {E(B)}\cdot[1-F(p_2)]
\end{equation*}
So we have:
\begin{equation*}
    h(p_2)\geq \frac{S}{m\cdot E(B)}=h(p^*)
\end{equation*}
which infers that $p_2\leq p^*$.

So in general, we have that $p_2\leq p^*\leq p_1$.
\end{proof}

One natural question to the duopoly market is that whether the company acting as a follower would
merge with the leading company in the market. To answer it, we follow the conventional way of
analyzing the total revenue and social welfare under competition and monopoly. At first glance, it
seems that allowing the search engine with better technology to monopolize the market would generate
higher total revenue and social welfare since it can provide better service for both advertisers and
end users. However, it turns out the answer depends on the specific parameters of participants in the
market. We summarize the comparison results of total revenue and social welfare under competition and
monopoly in the following theorem.

\begin{theo}\label{theo:revunue_SW_duo_vs_mono}
Whether monopoly would bring in higher total revenue and social welfare than competition
\emph{depends on} the specific parameters of advertisers in the advertising systems.
\end{theo}
We can prove this theorem by constructing the counter-examples \ref{expl:revenue} and \ref{expl:SW}
as follows.

\begin{expl}\label{expl:revenue}
Suppose there are two advertisers $\{1,2\}$ participating in the advertising system. The value,
budget and discount factor of each advertiser are as follows: $v_1=1$, $B_1=2$, $\rho_1=1$ and
$v_2=4$, $B_2=2$, $\rho_2=0$. The total supply of advertising opportunity is $S=1$.

Under monopoly, the optimal price $p^*=2$ and the maximal revenue is $R(p^*)=p^*\cdot S=2$. For any
$p<p^*$, the corresponding revenue would be $R(p)=p\cdot S<p^* \cdot S=2$; for any $p>p^*$, since
$v_1<p$ which means advertiser 1 would not attend the system, $R(p)$ is upper bounded by the budget
of advertiser 2, i.e., $R(p)\leq B_2=2$. This analysis proves the optimality of $p^*$ and $R(p^*)$.

Under competition, advertiser 1 would choose engine 2 since $\rho_1=1$ and advertiser 2 would choose
the other engine since $\rho_2=0$. We equally divide the supply into two parts: $S_1=S_2=0.5$. Now in
engine 1, the optimal price is $p_1=4$ and $R_1=p_1\cdot S_1=2$; in engine 2, the optimal price is
$p_2=1$ and $R_2=p_2\cdot S_2=0.5$. And the price ratio $p_2/p_1=0.25$ is less than $\rho_1$ and
greater than $\rho_2$.

Therefore in this example, the competition would bring in even higher total revenue ($R_1+R_2=2.5$)
than the monopoly ($R(p^*)=2$).
\end{expl}

\begin{expl}\label{expl:SW}
There are still two advertisers in the system, with the following parameters: $v_1=2$, $B_1=0.75$,
$\rho_1=0$ and $v_2=4$, $B_2=0.25$, $\rho_2=1$. The total supply of advertising opportunity is still
$S=1$.

Under monopoly, the optimal price $p^*=1$ and the allocation vector is $(q_1,q_2)=(0.75, 0.25)$. Thus
the social welfare would be $SW=v_1q_1+v_2q_2=2.5$.

Under competition, advertiser 1 would choose engine 1 since $\rho_1=0$ and advertiser 2 would choose
the other engine since $\rho_2=1$. We still divide the supply into $S_1=S_2=0.5$. Now in engine 1,
the optimal price is $p_1=1.5$ and the social welfare $SW_1=v_1q_1=v_1S_1=0.75$; in engine 2, the
optimal price is $p_2=0.5$ and $SW_2=v_2q_2= v_2 S_2=2$. And the price ratio $p_2/p_1=1/3$ is greater
than $\rho_1$ and less than $\rho_2$.

Therefore in this example, the competition would bring in even higher social welfare
($SW_1+SW_2=2.75$) than the monopoly ($SW=2.5$).
\end{expl}

Theorem \ref{theo:revunue_SW_duo_vs_mono} shows that there is no common conclusion on whether the
existence of an inferior company (or product) in the market would raise or drive down the social
welfare (or total revenue). Our observation here based on the particular search engine competition
model coincides with the finding in the recent paper \cite{QWang} that the viability of
\emph{differentiated services} scheme depends on the specific characteristics of users in the system.
The services provided by search engine 1 and engine 2 can be regarded as the 1st and 2nd class
services in \cite{QWang} where the 1st class is usually charged higher price than the 2nd (analogous
to our proposition that $p_1^{NE}\geq p_2^{NE}$).

Recall that our conclusions above are based on the \emph{ex post} perspective which includes all
possible instances of the competitive market. To show the more general \emph{ex ante} results under
common parameter setting of participants, we conduct the simulation in the next section.

\section{Simulation Results and Observations}\label{sec:5:simul}
In this section we present some simulation results showing the effects of different parameters in our
model. There are four major criteria we would like to explore in the model:
\begin{enumerate}[({a}.1)]
  \item \emph{Prices}: We would like to compare the equilibrium prices of both engines with the
  monopoly price if there is only one search engine dominates the market. In the following section we denote $(p_1,p_2)$ as the duopoly prices and $p_M$ as the monopoly price.
  \item \emph{Revenues}: It would be intriguing to study the comparative results of total revenues
  under competition and monopoly. The gap between revenues under competition and monopoly would serve
  as a signal of whether the leading company would like to propose a merger or acquisition to its
  competitor. A huge gap would infer that reaching certain cooperation agreement between the two competitors would
  significantly promote the revenues for both companies.
  \item \emph{Aggregate Utility of Advertisers}: We compare the aggregate utility of advertisers to
  see whether monopoly would be detrimental to the interest of advertisers, and if so, how severe
  the loss would be. In particular, we examine the aggregate utility for brand advertisers who
  benefits from the relatively lower price of the inferior search engine in the duopoly market.
  \item \emph{Social Welfare}: Social welfare can be regarded as the \emph{realized} value of
  advertisers and is the benchmark for addressing the interest of the community as a whole. Under
  competition, the social welfare is computed according to the following equation:
  \begin{equation}\label{eq:SW:simu}
    SW=\sum_{i\in\mathcal{I}_1}v_i q_i+\sum_{i\in\mathcal{I}_2}\rho_i v_i q_i.
  \end{equation}
  where $q_i$ is amount of supply allocated to advertiser $i$.
\end{enumerate}

In the following, we carry out a set of simulation to investigate the comparative results under
different parameter settings. For each simulation setting, we randomly generate 5000 instances of
parameters and calculate the average value of each criterion. The expected values from \emph{ex ante}
perspective can then be approximated by the average values of large amounts of \emph{ex post}
instances.

\textbf{(1). Baseline Setting:}

We consider two search engines equally dividing the market and the total supply is normalized to
unity. Thus the supply of either search engine is $S_1=S_2=0.5$. Advertisers' values are uniformly
distributed over $(18,20)$, and their budgets are also drawn from uniform distribution with
expectation $E(B)=4$. Discount factors of advertisers are uniformly distributed over $(0.5, 0.9)$.
Therefore there would be expectedly one half of advertisers with discount factors larger than the
average value $E(\rho)=0.7$, which we define as the \emph{brand advertisers}.

\begin{figure}[tbhp]
  \centering
  \subfloat[Prices]{
    \label{simu:fig1_a}
    \begin{minipage}{0.45\textwidth}
        \centering
        \includegraphics[width=\textwidth]{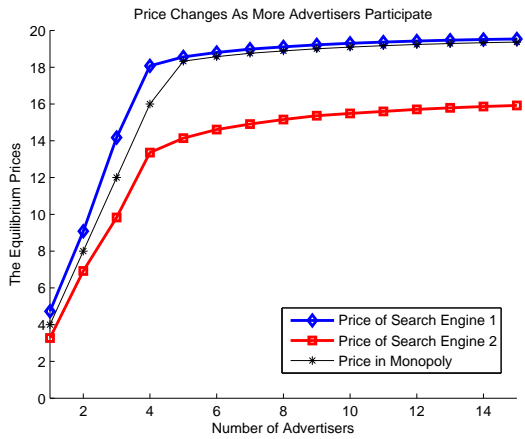}\\
    \end{minipage}
  }
  \subfloat[Revenues]{
    \label{simu:fig1_b}
    \begin{minipage}{0.45\textwidth}
        \centering
        \includegraphics[width=\textwidth]{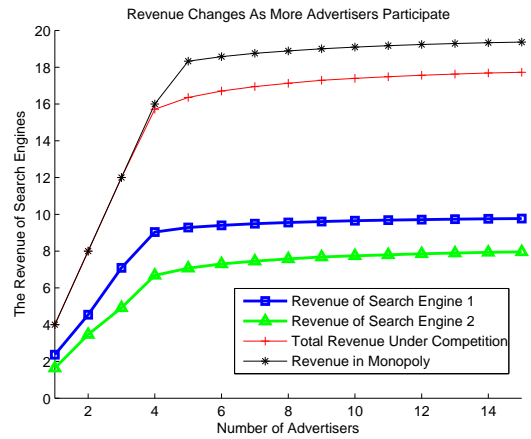}\\
    \end{minipage}
  }\\
  \subfloat[Aggregate Utility of Advertisers]{
    \label{simu:fig1_c}
    \begin{minipage}{0.45\textwidth}
        \centering
        \includegraphics[width=\textwidth]{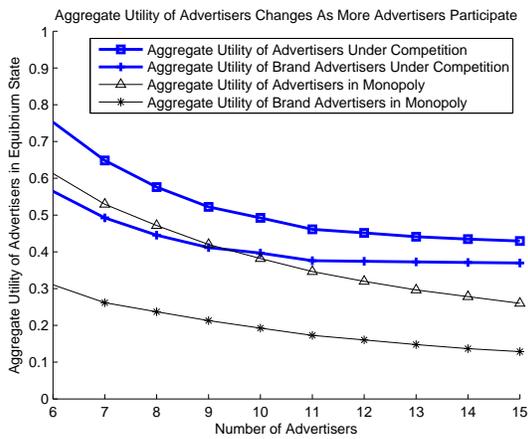}\\
    \end{minipage}
  }
  \subfloat[Social Welfare]{
    \label{simu:fig1_d}
    \begin{minipage}{0.45\textwidth}
        \centering
        \includegraphics[width=\textwidth]{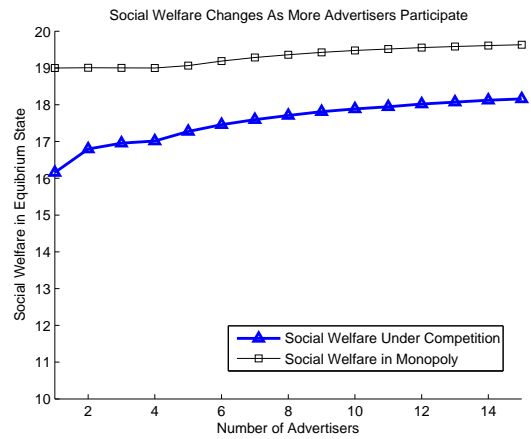}\\
    \end{minipage}
  }
  \caption{Baseline Setting}\label{simu:fig1}
\end{figure}

The simulation results under baseline setting are presented in figure \ref{simu:fig1}. We can make
the following observations from figure \ref{simu:fig1}(a)-(d):

\begin{enumerate}[1)]
  \item As the number of advertisers increases, the prices, revenues and social welfare would all get
  raised except the aggregate utility of advertisers. This is because as more advertisers
  participate, the demand for the limited supply would get boosted, which would finally drive up the unit
  price per supply and raise the revenue of search engines. As the price rises, the utility of
  advertisers would keep decreasing as seen in figure \ref{simu:fig1}(c). The social welfare can
  still be improved since when more advertisers appears, only those advertisers with higher values can
  stay and be allocated with certain amount of supply. Thus the \emph{realized} values of advertisers
  would be larger and the social welfare get enhanced.
  \item After the number of advertisers reaches about five, the growth of prices and revenues seems
  saturated: more advertisers would not bring evident enhancement in prices and revenues. This can be
  derived from our parameter setting: $E(B)/E(v)=4/19\thickapprox 0.2$ is the approximate amount of demand
  for each advertiser, and since the total supply is one, in expectation it would be
  sufficient for five advertisers to consume all the supply.
  \item Figure \ref{simu:fig1}(a) corresponds with theorem \ref{theo:price_duo_vs_price_mono} that
  the monopoly price is smaller than the duopoly price $p_1$ of engine 1 and larger than price $p_2$ of engine 2.
  We further notice that the monopoly price is actually very close to $p_1$ but $p_2$ is much smaller
  than $p_1$. This is because the monopoly engine and engine 1 in competition face advertisers with
  the same distribution of values. Recall that value is the maximal willingness to pay for
  advertisers, thus when there are too many advertisers competing with each other, the price would
  approach the maximal possible value, which is $20$ according to the distribution range. However,
  for search engine 2, the \emph{actual values} of advertisers are the \emph{original values} discounted by $\rho$.
  The lower $\rho$ is, the larger the gap between $p_1$ and $p_2$ would be.
  \item Figure \ref{simu:fig1}(b) shows that revenue of search engine 1 is larger than that of engine
  2. This can be easily deducted since the revenue of each engine is $R_1=p_1\cdot S_1$, $R_2=p_2\cdot S_2$ and we have $p_1>p_2$, $S_1=S_2$.
  As we have mentioned, the monopoly price $p_M$ is approximately equal to $p_1$. Therefore the
  monopoly revenue can be denoted as follows:
  \begin{eqnarray}
    R_M &=& p_M\cdot S\thickapprox p_1\cdot S=p_1\cdot S_1 + p_1\cdot S_2 \nonumber\\
        &=& R_1+\frac{p_1}{p_2}R_2=R_1+\frac{R_2}{\rho^*}>R_1+R_2         \label{eq:revenue_mono:simu}
  \end{eqnarray}
  where $\rho^*$ is the discount factor of the indifferent advertisers which is always less than one.
  This inequality explains the gap between total revenue under competition and monopoly in figure
  \ref{simu:fig1}(b).
  \item The utility of advertisers depends on two factors: the value and the price. Compared with the
  monopoly, under competition a portion of advertisers could enjoy a relatively lower price which would
  result in higher utility; at the same time, due to the effect of $\rho$, the values of advertisers
  in engine 2 get discounted which would cause lower utility. When the positive factor of lower price
  dominate the negative factor of lower value, the utility under competition would be greater and vice
  versa. Since in our baseline setting we set a relatively large $\rho$, the negative factor would be
  small and advertisers in engine 2 can benefit from the lower price. This conjecture can be verified
  in figure \ref{simu:fig1}(c). Since in average \emph{brand advertisers} account for half of all
  advertisers, in monopoly the utility of brand advertisers is always half of the utility of all
  advertisers as shown in figure \ref{simu:fig1}(c). However, under competition the brand
  advertisers would benefit more than the rest advertisers since they have higher discount factors
  which means lower negative effect on values but confronting a lower price in engine 2.
  \item Figure \ref{simu:fig1}(d) indicates that the social welfare under competition is lower than
  that under monopoly since the realized values in equation (\ref{eq:SW:simu}) get discounted due to the factor
  $\rho$.
\end{enumerate}

\textbf{(2). Effect of Supplies}

We now change the supplies to $S_1=0.9$ and $S_2=0.1$ while all other parameters remain the same. The
simulation results are presented in figure \ref{simu:fig2}.

In this setting one search engine plays the leading role in the market and the follower can only take
a small fraction of market share. This assumption is more realistic considering the current dominant
position of Google in most areas of the world.\footnote{In the United States around 72 percent of the
total search volumes are conducted on Google while Yahoo and Bing jointly account for about 25
percent. Source from ``Top 20 Sites \& Engines,'' Hitwise, May 20, 2010
(\url{http://www.hitwise.com/us/datacenter/main/dashboard-10133.html}). In some other countries like
France, UK and Germany, Google even possessed a market share of over 90 percent. Source from the
article ``Google's Market Share in Your Country,'' March 13, 2009
(\url{http://googlesystem.blogspot.com/2009/03/googles-market-share-in-your-country.html}).} Figure
\ref{simu:fig2}(a) shows little difference with the corresponding price curves in figure
\ref{simu:fig1}(a) since all prices approach the maximal possible value when there are sufficient
number of advertisers in the market. Since the supply of engine 2 decreases, the revenue of engine 2
which is denoted by $R_2=p_2\cdot S_2$ would also drop. According to equation
(\ref{eq:revenue_mono:simu}), the gap of the total revenues is
$R_M-(R_1+R_2)\thickapprox(\frac{1}{\rho^*}-1)R_2$. Therefore when $R_2$ is small, the gap would also
be negligible. This is verified by figure \ref{simu:fig2}(b). This result also demonstrates that even
when the follower takes a small portion of market share and provides service of relatively lower
quality, it can still make nontrivial profit through competition and survive in the market. In figure
\ref{simu:fig2}(c), the gain of aggregate utility under competition is tiny compared with that under
monopoly since engine 2 can only provide limited supply (10\% of the total supply) and only very few
of advertisers can take advantage of it. Figure \ref{simu:fig2}(d) shows that the loss of social
welfare under competition is very small since only 10\% of the total supply is of lower quality,
i.e., the second term in equation (\ref{eq:SW:simu}) is insignificant.

\begin{figure}[ht]
  \centering
  \subfloat[Prices]{
    \label{simu:fig2_a}
    \begin{minipage}{0.45\textwidth}
        \centering
        \includegraphics[width=\textwidth]{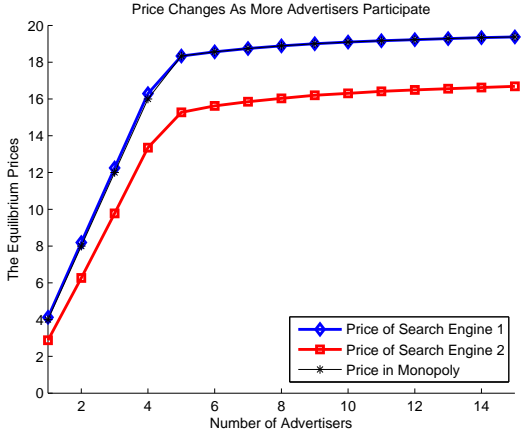}\\
    \end{minipage}
  }
  \subfloat[Revenues]{
    \label{simu:fig2_b}
    \begin{minipage}{0.45\textwidth}
        \centering
        \includegraphics[width=\textwidth]{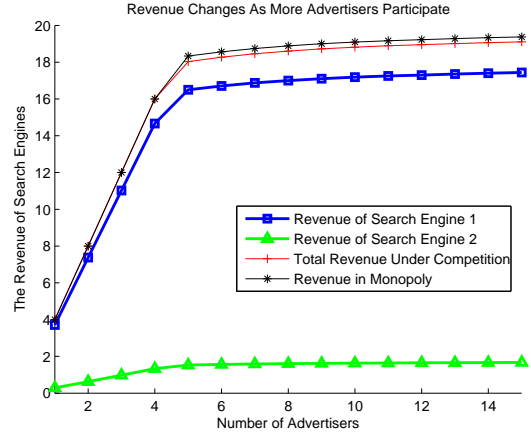}\\
    \end{minipage}
  }\\
  \subfloat[Aggregate Utility of Advertisers]{
    \label{simu:fig2_c}
    \begin{minipage}{0.45\textwidth}
        \centering
        \includegraphics[width=\textwidth]{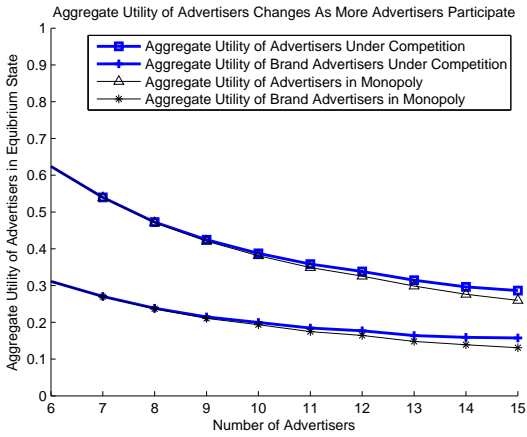}\\
    \end{minipage}
  }
  \subfloat[Social Welfare]{
    \label{simu:fig2_d}
    \begin{minipage}{0.45\textwidth}
        \centering
        \includegraphics[width=\textwidth]{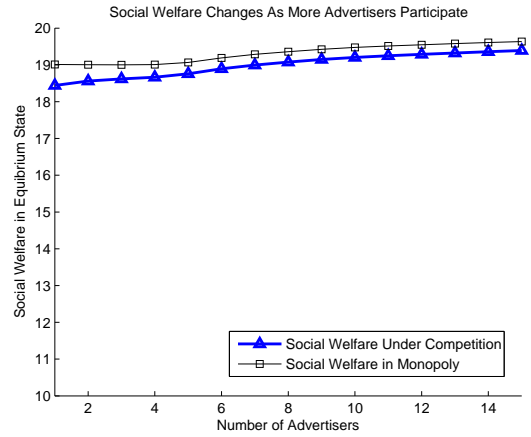}\\
    \end{minipage}
  }
  \caption{When Supplies Change To $S_1:S_2=9:1$}\label{simu:fig2}
\end{figure}

\textbf{(3). Effect of Discount Factors}

We now turn to investigate the effect of technology gap between two search engines. Let the discount
factors be drawn uniformly on $(0.1, 0.5)$ and all other parameters are the same as the baseline
setting.

As figure \ref{simu:fig3}(a) shows, the equilibrium price $p_1$ of engine 1 and the monopoly price
$p_M$ are still close to the maximal value of advertisers and $p_2$ approaches the maximal discounted
value $\rho v$. As $\rho$ decreases, the price of engine 2 also diminishes compared with $p_2$ in
figure \ref{simu:fig1}(a) and \ref{simu:fig2}(a).

Since the revenue $R=p\cdot S$ and from above analysis we know $p_1$ has little change and $p_2$
diminishes, the revenue $R_1$ of engine 1 would stay almost the same while $R_2$ reduces as shown in
figure \ref{simu:fig3}(b). Since we have also mentioned that the gap between the total revenues
approximates to $R_M-(R_1+R_2)\thickapprox(\frac{1}{\rho^*}-1)R_2$, when $\rho$ is small, the gap
would become larger. This can be easily seen by comparing the corresponding revenue curves in figure
\ref{simu:fig1}(b) and figure \ref{simu:fig3}(b).

The aggregate utilities of advertisers in monopoly are the same in figure \ref{simu:fig1}(c) and
figure \ref{simu:fig3}(c); however, the aggregate utilities under competition in figure
\ref{simu:fig3}(c) is much smaller than those in figure \ref{simu:fig1}(c) due to the negative effect
of $\rho$ on advertisers' values. In figure \ref{simu:fig3}(c) we see that there is certain
intersection between utility under competition and monopoly. When there are only a few of
advertisers, the prices are still low and the main factor affecting utility is the value. Under
competition the existence of $\rho$ would drive down advertisers' values and thereby results in lower
aggregate utility. As the number of advertisers increases, the monopoly price $p_M$ would approach
the maximal value $v_{max}=20$ and the aggregate utility would gradually reduce to zero. However,
even when $p_1$ and $p_M$ approach $20$, the rest of advertisers whose discount factors are larger
than the equilibrium price ratio $\rho^*=p_2/p_1$ can still obtain nontrivial utility which can be
approximated as $(\bar{\rho} v-p_2)S_2\thickapprox (\bar{\rho} v-\rho^* v)S_2$ ($\bar{\rho}$ denotes
the average value of advertisers with discount factors larger than the price ratio $\rho^*$).

Figure \ref{simu:fig3}(d) displays the hug gap between social welfare under competition and monopoly.
Since half of the supply ($S_2=0.5$) is allocated to advertisers with discount factors less than
$0.5$, the realized values in equation (\ref{eq:SW:simu}) would become significantly smaller than
social welfare under monopoly.

\begin{figure}[ht]
  \centering
  \subfloat[Prices]{
    \label{simu:fig3_a}
    \begin{minipage}{0.45\textwidth}
        \centering
        \includegraphics[width=\textwidth]{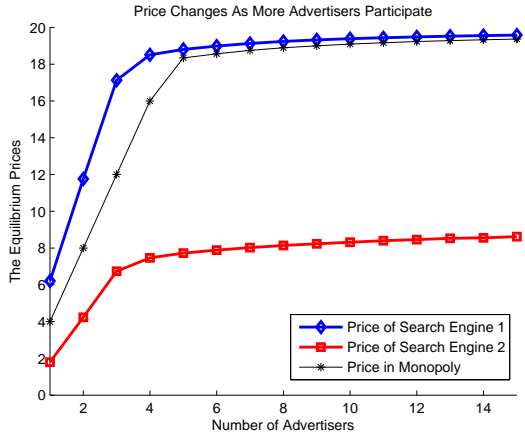}\\
    \end{minipage}
  }
  \subfloat[Revenues]{
    \label{simu:fig3_b}
    \begin{minipage}{0.45\textwidth}
        \centering
        \includegraphics[width=\textwidth]{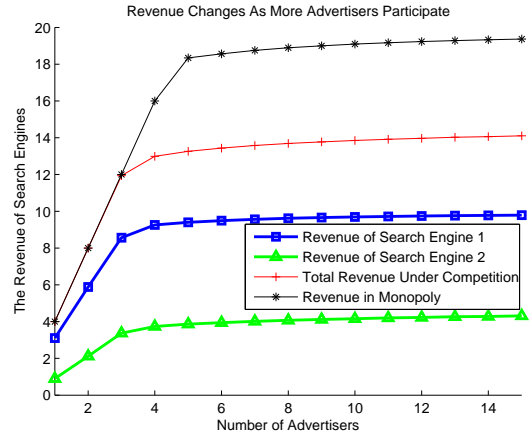}\\
    \end{minipage}
  }\\
  \subfloat[Aggregate Utility of Advertisers]{
    \label{simu:fig3_c}
    \begin{minipage}{0.45\textwidth}
        \centering
        \includegraphics[width=\textwidth]{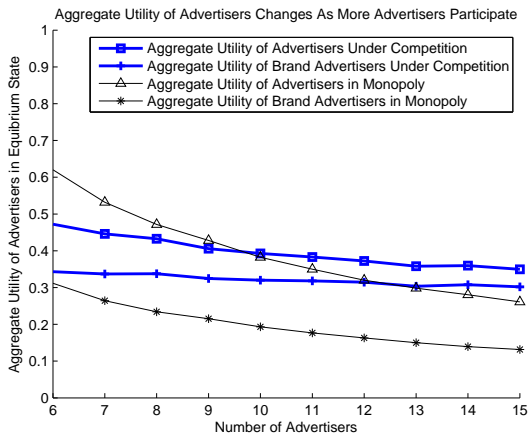}\\
    \end{minipage}
  }
  \subfloat[Social Welfare]{
    \label{simu:fig3_d}
    \begin{minipage}{0.45\textwidth}
        \centering
        \includegraphics[width=\textwidth]{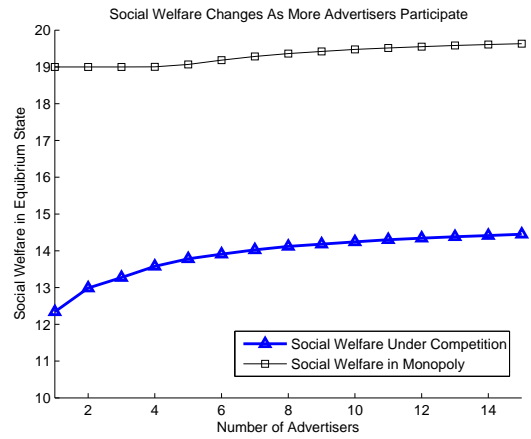}\\
    \end{minipage}
  }
  \caption{When $\rho$ Belongs to $(0.1, 0.5)$}\label{simu:fig3}
\end{figure}

\section{Conclusion and Future Work}\label{sec:6:conclu}
We propose an analytical framework to model the interaction of publishers, advertisers and users
under monopoly and duopoly scenarios. For monopoly market, we can give the analytical results of
price and revenue for both ex ante and ex post case. For duopoly market, we formulate a three-stage
dynamic game to model the search engines' competition for both users and advertisers and prove the
existence of Nash equilibrium from \emph{ex post} perspective. To see the \emph{long-term} effect of
competition from \emph{ex ante} perspective, we carry out computer simulations for different settings
of participants' parameters. The comparative results of revenues and social welfare under competition
and monopoly are then presented and discussed extensively in the paper. Our analysis could provide
some insight in regulating the search engine market and protecting the interests of advertisers and
end users. Although the cooperation between search engines can probably bring more total revenues,
advertisers and users may be averse to such plan which eliminates their freedom to choose from
diverse services provided by different search companies.

There are several possible ways to extend our work. Throughout our paper, we implicitly assume that
advertisers would reveal their true parameters such as values and budgets to the search engines.
Since our framework is by no means strategy-proof, how would rational advertisers' strategies affect
our conclusions would be an interesting question for further investigation. Another non-trivial
problem is how to associate our analytical result of revenue from \emph{one} particular keyword with
the practical revenue of an industrial search company which gathers from \emph{numerous} keywords
queried by end users everyday. Besides, to be in line with the practical advertising system nowadays,
we will consider incorporating the quality factor of advertisement for the revenue-ranking rule as
well as the generalized second-price auction prevailing in major search engines. Finally, it would be
intriguing to extend our model for multiple search engines scenario besides the basic duopoly
scenario.

%=============================================================================

\end{document}